\documentclass[11pt]{article}
\usepackage[top=1in, bottom=1in, left=1in, right=1in]{geometry}

\usepackage{amsmath}
\usepackage{amssymb}
\usepackage{parskip}
\usepackage{algorithm}
\usepackage{subfig}
\usepackage{graphicx} 
\usepackage{amsthm}
\usepackage{physics}
\usepackage{booktabs}
\usepackage{hyperref}
\usepackage{xcolor}
\usepackage{multirow}
\usepackage[numbers]{natbib}
\usepackage{mathpazo}

\usepackage{amsmath}
\usepackage{algorithm}
\usepackage{subfig}
\usepackage{graphicx} 
\usepackage{amsthm}
\usepackage{physics}
\usepackage{booktabs}
\usepackage{hyperref}
\usepackage{xcolor}
\usepackage{multirow}

\usepackage{url}

\newtheorem{theorem}{Theorem}[section]

\newtheorem{proposition}[theorem]{Proposition}

\newcommand{\E}{\mathbb{E}}
\newcommand{\R}{\mathbb{R}}

\renewcommand{\P}{\mathbb{P}}

\newcommand{\set}[1]{\{#1\}}

\newcommand{\cg}[1]{\textcolor{black}{#1}}
\newcommand{\nrev}[1]{\textcolor{black}{#1}}

\title{Addressing Strategic Manipulation Disparities in Fair Classification}

\usepackage{authblk}

\author{Vijay Keswani}
\author{L. Elisa Celis}
\affil{Yale University}
\date{}

\begin{document}

\maketitle
\begin{abstract}
In real-world classification settings, such as loan application evaluation or content moderation on online platforms, individuals respond to classifier predictions by strategically updating their features to increase their likelihood of receiving a particular (positive) decision (at a certain cost).
Yet, when different demographic groups have different feature distributions or pay different update costs,
prior work has shown that individuals from minority groups often pay a higher cost to update their features.
Fair classification aims to address such classifier performance disparities by constraining the classifiers to satisfy statistical fairness properties.
However, we show that standard fairness constraints do not guarantee that the constrained classifier reduces the disparity in strategic manipulation cost.
To address such biases in strategic settings and provide equal opportunities for strategic manipulation, we propose a constrained optimization framework that constructs classifiers that lower the strategic manipulation cost for minority groups.
We develop our framework by studying theoretical connections between group-specific strategic cost disparity and standard selection rate fairness metrics (e.g., statistical rate and true positive rate).
Empirically, we show the efficacy of this approach over multiple real-world datasets.

\end{abstract}

\section{Introduction}
In prediction/classification settings, the goal is to develop automated models that accurately predict class labels 
using the available demographic and task-specific features of individuals.
The use of predictive models in many real-world applications, however, impacts the features of the underlying population. One direct way this happens is when individuals take steps to update their features to potentially obtain a different prediction in the future.
In binary classification, where positive class labels can denote \textit{success} for a given task, individuals who have been negatively classified will attempt to update their features in a manner that increases their likelihood of receiving a positive decision in the future.

There are numerous examples of such individual behavior in response to institutional decisions.
Consider the setting of loan applications, where the features are individuals' demographics, annual income, credit history, number of dependents, the number of open credit lines, etc.
The class label to be predicted is whether an individual will default on a loan or not.
The number of open credit lines is a feature that is often positively correlated with the class label and individuals can increase their likelihood of positive loan application (or increase their \textit{credit score}) by opening more credit lines (for the sake of simplicity, assume all the other features are unchanged; in real-world scenarios, there will be simultaneous dependence on other variables as well here, e.g. whether the individual has been regular with their payments or not).
However, opening credit lines requires additional investment on the part of the individuals \cite{citron2014scored}.
Another example is social media websites and online platforms. Even without a complete understanding of a platform's recommendation system, users nevertheless attempt to intervene in different ways to exercise control over the platform's algorithms \cite{van2018networks}.
For example, content moderation tools used in social media platforms flag objectionable posts, which are then suppressed by the recommendation system to ensure low visibility \cite{gillespie2022not} (often unfairly targeting minority voices \cite{haimson2021disproportionate,zeng2022content}).
Users, in this case, curate and modify their content to work around the platform's decision
\cite{burrell2019users,van2018networks}.
Beyond content moderation, strategic manipulation can allow users to avoid harassment, as seen in the case of Twitter \cite{burrell2019users}.
Evidence of users' attempts to exercise control over an online platform's algorithms has similarly been 
observed in ride-hailing apps like Uber and Lyft \cite{mohlmann2017hands}.
A final example is the setting of college admissions, where the features are individuals' demographics, school academic records, extra-curricular records, and scores from standardized tests like GRE.
The class label to be predicted is the likelihood of academic ``success'' to determine college admissions.
In this case, while higher scores for standardized tests increase the chances of a successful college application,
students can take these tests multiple times and submit only the highest scores.
Nevertheless, there is an additional investment required as every additional test attempt involves monetary and time expenses \cite{vigdor2003retaking}.
\cg{Feature manipulations of these kinds can also take the form of positive steps taken by individuals to improve their features (e.g., investing additional time in test preparation)
\cite{kleinberg2020classifiers,alon2020multiagent}
and/or provide individuals with the agency to address model decisions \cite{ustun2019actionable,venkatasubramanian2020philosophical}.
}

The above-described process involves two main players: the institution constructing a classifier and the individuals reacting to the classifier.
While the institution's goal is to minimize prediction error (or maximize a certain measure of utility), individuals react to the classifier predictions by \textit{strategically manipulating} their features to achieve a positive classification.
In these {strategic settings}, often due to historical biases,
the classifier employed by the institution 
can pose relatively higher costs for strategic manipulation (i.e., increased costs to improve their feature values) for individuals from minority groups (e.g., race and gender minorities).
Prior work has observed such disparities in settings where the datasets used for training the classifier encode social biases or when minority groups pay larger costs to update their features
\cite{milli2019social, hu2019disparate}.
For example, in the case of loan applications, historical discrimination against African Americans in financial aspects often deters them from seeking new credit lines
\cite{cnbc2019}.
In the case of social media platforms, content moderation tools exhibit bias against minority groups, for example, by reducing the visibility of posts by advocates from minority groups \cite{tiktok2020, salty2020, haimson2021disproportionate} or by using biased sentiment analysis tools
\cite{kiritchenko2018examining,diaz2018addressing}; these biases lead to greater hurdles for these groups to make their voices heard.
\nrev{Similarly, for graduate school admissions, \citet{wilson2020predicting} 
revealed limitations of GRE and UGPA scores in predicting graduate school success for Black students. Using these scores without considering the racial disparities can create a higher admission barrier for Black students.
}
These biases are a result of negative stereotypes and/or historical lack of opportunities for minority groups and classifiers that inherit such biases can further propagate them.
{In the presence of these biases
in the predictions of trained classifiers, one can ask whether an institution can construct classifiers that provide equal opportunities for strategic manipulation to all groups and address the systemic disparities in investment required to improve their outcomes.
}
Strategic manipulation opportunities often serve as mechanisms to provide individuals with \textit{recourse} or \textit{agency} against biased institutional decisions \cite{venkatasubramanian2020philosophical}.
\nrev{As such, equalizing manipulation opportunities will ensure that majority groups do not solely take advantage of effective strategic manipulations and provide similar power to minority groups to address classifier decisions.}

Fairness-constrained classification attempts to address such disparities in classifier performance by constraining the classifier to satisfy certain statistical fairness properties.
For example, when constraining with respect to \textit{statistical rate}, the classifiers are constrained to have an almost-equal selection rate for all groups \cite{zafar2017fairness, dwork2012fairness, celis2019classification, rezaei2020fairness,agarwal2018reductions}.
Similarly, when constraining with respect to \textit{equalized odds}, the classifiers are constrained to have equal false positive and true positive rates for all groups \cite{Hardt2016EqualityOO, celis2019classification, rezaei2020fairness}.
However, these fairness metrics and constraints operate in a \textit{static} manner and do not take into account the response of the individuals to the classifier predictions or the disparity in costs that different groups pay for updating their features.
Even though fairness constraints encourage the increased selection of minority group individuals, existing dataset biases or update cost disparities can still disproportionately affect the negatively-classified individuals in the minority group.
\textit{Correspondingly, the primary question we investigate is the following: Do constraints that use standard static fairness metrics lead to classifiers that reduce strategic manipulation costs for minority groups?
}

\vspace{0.5em}
\noindent
\textbf{Our Contributions.}
We first theoretically study the relationship between standard selection rate fairness metrics (like statistical rate and true positive rate disparity) and the disparity in strategic manipulation costs between majority and minority groups when only one-dimensional features and group membership of individuals are provided (Section~\ref{sec:one_dim}).
Our analysis shows that threshold-based classifiers that have an equal selection rate for all groups can still have higher strategic manipulation costs for the disadvantaged groups when
feature distributions or cost functions differ across groups.
(Theorem~\ref{thm:main}, \ref{thm:main_diff_costs}).
Prior works on strategic cost disparities only demonstrated that this disparity can be large in unconstrained settings
\cite{milli2019social, hu2019disparate}.
Our analysis demonstrates that even ``fair'' classifiers that are constrained using selection rate fairness metrics can still have large strategic cost disparities. 
To address this bias, we bound the strategic cost disparity using the statistical properties of the classifier and the cost function.
Using these bounds as constraints, we construct classifiers that have both low selection rate disparity and low strategic cost disparity.
We also extend the results to multi-dimensional settings when the classifier is linear and the cost function is linear or quadratic (Section~\ref{sec:multi_dim}).
\nrev{The primary technical challenge we face in proving our results is accounting for all factors that result in strategic cost disparity. As we discuss in Section~\ref{sec:one_dim}, this disparity can arise due to multiple reasons, such as unequal group selection rates, variation in cost functions across groups, and ``distance'' of negatively classified individuals from classifier thresholds. Correspondingly, our results quantify the relationship between  the strategic cost incurred by each group and the group’s selection rate using all relevant factors, including bounds on the cost function gradient and other related empirical properties of the classifier. 
Using these bounds, we can construct appropriate classifiers that minimize strategic cost disparity.
}
Our theoretical results are complemented by empirical analysis on two real-world financial datasets: the FICO credit dataset \cite{Hardt2016EqualityOO} and the Adult income dataset \cite{ding2021retiring} (Section~\ref{sec:empirical}).
For both datasets, we show that fair classification with our proposed constraints leads to lower manipulation costs for the minority group.

\vspace{0.5em}
\noindent
\textbf{Related Work.}
Studies by \citet{milli2019social} and \citet{hu2019disparate} first analyzed strategic manipulation cost disparities when feature distributions or cost functions are biased against minority groups.
However, their analysis is limited to classifiers that optimize institution utility; in contrast, we also study classifiers that optimize utility subject to standard fairness constraints.
\citet{estornell2021unfairness} and \citet{braverman2020role}, on the other hand, assess classifier fairness in strategic settings using only selection rate fairness metrics.
\citet{estornell2021unfairness} observed that statistical parity or equalized odds constrained classifiers become less ``fair'' (with respect to the same metrics) than unconstrained classifiers due to 
strategic manipulations.
\citet{braverman2020role} study the impact of randomness on classifiers trained in strategic settings and
propose the use of \textit{noisy features} to address selection rate disparities in the outputs of these classifiers.
Like our work, both these papers evaluate the impact of fair classification in strategic settings; however, they analyze the fairness of final individual outcomes using only selection rate metrics and do not consider the costs disparity across groups.
Similar to strategic updates, \citet{ustun2019actionable} consider the notion of \textit{actionable recourse} and provide tools to minimize recourse cost for linear classifiers. However, their work does not aim to address recourse cost disparities.
\citet{gupta2019equalizing,von2022fairness} extend this line of work to study recourse disparities in classification; however, their models only handle settings where cost functions are the same for all groups.
As noted in multiple prior studies \cite{cohen2011credit, venkatasubramanian2020philosophical}, minority group individuals often pay larger costs to update their features, which then leads to
recourse disparities.
Our framework is, hence, more generic (than \cite{gupta2019equalizing,von2022fairness}) as it tackles both cost and feature disparities.

Static fairness constraints in non-strategic settings, that compare the selection rate of majority and minority groups, have been extensively studied in the context of constructing fair classifiers
\cite{kamishima2012fairness, dwork2012fairness, zafar2017fairness, kusner2017counterfactual, celis2019classification, zhang2018mitigating, agarwal2018reductions, rezaei2020fairness, zhang2019faht}.
For non-strategic settings, \citet{hu2020fair} show that selection rate-constrained classifiers may not improve the average quality of predictions received by the disadvantaged groups.
We extend this direction to analyze the impact of fair classification in strategic settings.
Recent work on strategic settings has also studied classifiers that are robust to strategic updates \cite{hardt2016strategic,kleinberg2020classifiers, chen2020learning,dong2018strategic,jagadeesan2021alternative, haupt2023recommending}. 
The analysis in these papers is primarily 
from the viewpoint of an institution maximizing its utility given information about individuals' behavior and these papers do not consider the fairness goal of reducing manipulation costs disparities with respect to protected attributes.
Our paper, instead, considers the individuals' perspective and addresses the cost disparities arising from group memberships.
While we look at the one-step feedback models,
\textit{performative prediction} algorithms model multi-step feedback settings
to construct classifiers that are stable over induced distributions 
\cite{perdomo2020performative, miller2021outside}. 
For ease of analysis, we limit our study to one-step feedback settings.
\section{Model formalization} \label{sec:model}

Let $x \in \mathcal{X} \subseteq \R^d$ denote the features of an individual in the population, $y \in \set{0,1}$ denote the true class label to be predicted and $z \in \mathcal{Z}$ denote the protected attribute (assumed to be binary for our current analysis).
We will use $\mathcal{D}$ to denote the underlying joint distribution of features, class labels, and protected attributes, and let $X, Y, Z$ denote the respective random variables.
We will work with threshold-based classifiers $f: \mathcal{X} \rightarrow \set{0,1}$ which set a threshold on the likelihood of any point achieving a positive class label 
\footnote{\cg{Any hypothesis class where the classifier output is a distribution over the labels (e.g., logistic regression, Naive Bayes and MLPs) can be represented using threshold-based classifiers and correspondingly used in our framework.} 
}.

\vspace{0.5em}
\noindent
\textbf{Strategic manipulations and individual cost functions.}
As mentioned earlier, individuals can update or manipulate their features at a certain cost after observing a classifier prediction.
Let $c : \mathcal{X}{\times}\mathcal{X} \rightarrow \R$ denote the cost function such that $c(x,x')$ is the cost paid by an individual to update their feature from $x$ to $x'$.
The subsequent utility gained by the individual from this update can be quantified as $u_x(f, x') := f(x'){-}c(x,x')$.
In this setting, the optimal feature update for an individual (in response to a classifier $f$) is captured by $\Delta_f(x) := \arg \max_{x'} u_x(f,x').$
Since we only aim to model updates that lead to improved classifier prediction, we will study cost functions that have 
no feature update cost if the individual is already positively classified.
In other words, individuals are \textit{rational} and aim to maximize their utility (this assumption is consistent with prior work on strategic settings \cite{hardt2016strategic,hu2019disparate}).

The institution's aim, in the unconstrained setting, is to minimize error w.r.t. a given loss function $\mathcal{L}$, i.e., find the classifier $f$ that minimizes $\E_{\mathcal{D}} [\mathcal{L}(f; X,Y)]$.
\nrev{When using unmanipulated data, this loss will be a proxy measure for  $\P_{\mathcal{D}}[f(X) = Y]$, while for manipulated data, this loss will be a surrogate for $\P_{\mathcal{D}}[f(\Delta_f(X)) = Y]$ (i.e., the standard accuracy measure in strategic classification
}\cite{hardt2016strategic,miller2021outside,levanon2021strategic}).
\nrev{Any common classification loss function can be used for $\mathcal{L}$; e.g., we use the log-loss function in some of our simulations. Other common loss functions, such as mean square loss, hinge loss, or regularized versions of these functions, can also be used with our framework.}
However, to incorporate fairness in this optimization program, we need additional fairness constraints.

\vspace{0.5em}
\noindent
\textbf{Ground truth label $y$.} 
Feature manipulations represent actions or changes that are under individual-level control, such that positively manipulating the relevant features can potentially lead to a change in the classifier decision for this individual. 
In a variety of real-world settings, individual actions to update features $x$ can either change their ground truth label $y$ or not affect their ground truth label depending on the nature of the update and the context.
In all cases, it is important to ensure that equal opportunities for manipulations are available across demographic groups.
However, in this paper, we primarily consider the settings where strategic manipulations are used as a recourse option to address unfair institutional decisions.
While our model and theoretical analysis can handle both settings (i.e., when manipulations change ground truth and when they don't affect ground truth), our empirical analysis will primarily focus on cases where ground truth remains unchanged due to strategic updates, as these updates capture recourse strategies.
This assumption is also consistent with other works on 
strategic or adversarial manipulations \cite{hardt2016strategic, estornell2021unfairness, goodfellow2018making}.

To see why it is important to study strategic manipulations as a recourse option we present a few examples where feature updates do not lead to a change in ground truth label $y$ but still provide valuable agency to individuals.

\noindent
\textit{-- Example (1):} On online platforms, costs associated with individuals' actions can be seen to depend on a variety of factors. Many studies have reported that Black activists face higher levels of censorship on social media platforms simply due to mentions of race-associated terms \cite{haimson2021disproportionate}. To counter this, such activists have to manipulate their posts (e.g., by changing certain words or using screenshots) to get around the automated moderation tools, paying a cost in terms of time and resources required for such manipulation. Note that, these actions do not change the ground truth label $y$ of the post (i.e., the post continues to remain non-offensive), yet the censored individuals have to take action and pay associated costs so that the automated system aligns with their ground truth label.

\noindent
\textit{-- Example (2):} In a lending situation, an individual’s credit score is an important factor when evaluating their loan application. However, many studies have shown that changes in credit scores are related to factors beyond an individual’s financial credibility. 
Take the following example from a CNBC article\footnote{\url{https://www.cnbc.com/select/paying-off-credit-card-debt-boosts-credit-score/}}: two people with equal annual income, equal credit card debt, and equal credit limit got the \$1200 stimulus check. The first one used the entire amount to pay off the credit card debt while the second one used \$600 towards paying off their debt and used \$600 for their savings. However, due to different \textit{credit utilization} rates, the first person’s credit score will be higher than the second person’s score. 
Both individuals in this case have the same income/resources and, if they both applied for a loan, they arguably would have a similar likelihood to pay back the loan (implying no significant changes in ground truth for default risk $y$). Yet, because the first person has a higher credit score, any decision-making policy that uses credit scores would prefer the first person for the loan. Hence, individual actions affect classifier decisions, sometimes independent of ground truth.

The above examples are settings where individuals can use strategic manipulations as a recourse option when they believe that the institution's decision is not correct.
These examples also make it clear that addressing incorrect decisions can require monetary or time investments by individuals.
In these examples and numerous other settings explored in prior work \cite{tiktok2020, diaz2018addressing, salty2020, wilson2020predicting}, strategic manipulations are used by individuals to exercise control over institutional decisions. 
We primarily analyze strategic manipulations in these recourse contexts 
because our analysis considers the perspective of the individuals and the agency provided to them to address unfair institutional decisions.
The institution can use a classifier that either simply maximizes their utility/accuracy or they can use one that is fair with respect to standard selection rate fairness metrics.
We evaluate the impact of such classifiers on individuals from different groups and the average cost they have to pay to positively manipulate their features based on their group membership. 

\vspace{0.5em}
\noindent
\textbf{Selection rate fairness metrics.}
The protected attribute $z \in \mathcal{Z}$
is the focus of our analysis of fairness.
Standard fair classification algorithms measure fairness using group-specific selection rates, either over the entire population or over certain subpopulations
\cite{celis2019classification, mitchell2021algorithmic}.
For any sub-population condition $\psi : \mathcal{X} \times \mathcal{Y} \rightarrow \{0,1\}$, the (conditional) selection rate of a classifier $f$ with respect to protected attribute group $z$ can be defined as $H_z(f, \psi) := \P_{\mathcal{D}}[f(X) = 1 \mid \psi(X,Y) = 1, Z=z]$.
With respect to this definition, the conditional selection rate fairness of $f$ can be quantified as 
$$H(f, \psi) := H_0(f, \psi) - H_1(f, \psi).$$

\noindent
If the condition is identity, i.e. $\psi(x, y) = 1$, then $H_z(f, \psi)$  simply measures the fraction of elements in group $z$ that are positively classified and $H(f, \psi)$ in this case is the standard statistical rate metric
\cite{dwork2012fairness}.
If the condition is $\psi(x, y) = \mathbf{1}(y{=}1)$, then $H_z(f, \psi)$ measures the true positive rate for group $z$, and $H(f, \psi)$ is the true positive rate disparity across the protected attribute groups
\cite{zafar2017fairness, Hardt2016EqualityOO}.
Using the above definition of $H$, all standard \textit{linear fairness metrics} considered in \citet{celis2019classification} can be represented in additive form.
\cg{When clear from context, we will use shorthand $\psi$ to denote $\psi(\cdot,\cdot)$.}

\vspace{0.5em}
\noindent
\textbf{Strategic cost disparity.}
The power of strategic manipulation can be different for different demographic groups,
which is the primary kind of bias we tackle in this paper.
\nrev{
As mentioned earlier, these biases can occur when the underlying distributions vary across groups
due to possibly different historical evolution trajectories followed by group-specific distributions  \cite{cnbc2019}, or when one group pays larger update costs than others for similar updates \cite{board2007report}.
}

For a classifier $f$, the expected cost incurred by individuals from group $Z=z$ can be measured using 
$ \E_{\mathcal{D}}[c(X, \Delta_f(X)) \mid Z=z]$, 
a quantity referred to as the \textit{social burden} for the group $z$ by  \citet{milli2019social}.
Therefore, one measure of fairness we can look at in this strategic setting is the following gap:
$\E[c(X, \Delta_f(X)) \mid Z = 0] - \E[c(X, \Delta_f(X)) \mid Z = 1]$.
Higher values ($> 0$) of this quantity imply that individuals from group 0, on average, pay a larger cost to strategically manipulate their features than individuals from group 1.
While the above measure evaluates the cost for all individuals in each group,
different contexts might require focusing on different sub-populations of individuals from each group.
E.g., in the recidivism risk assessment setting \cite{washington2018argue}, we may want to analyze the average cost paid by a low-risk individual from the minority group who has been deemed high-risk to overturn the classifier decision. In this case, the expected cost $\E[c(X, \Delta_f(X)) \mid Y = 1, Z = z]$ is more relevant ($Y = 1$ denotes low-risk). 
Hence, in general, for any classifier $f$, we can define the social burden for any group $z$ with respect a given sub-population condition 
$\psi : \mathcal{X} \times \mathcal{Y} \rightarrow \set{0,1}$ as $G_z(f, \psi) := \E[c(X, \Delta_f(X)) \mid \psi(X,Y) = 1, Z = 0]$ and, correspondingly, define the \textit{social burden gap} as
\begin{align*}
G(f, \psi) := G_0(f, \psi) - G_1(f, \psi) .
\end{align*}

\noindent
\nrev{Classifiers that equalize manipulation costs ($G(f, \cdot)=0$) ensure that all groups have similar manipulation power.
In certain cases, we might even require $G(f, \cdot)$ less than 0 to counter historical inequalities faced by disadvantaged groups.
Hence, our goal is to provide an optimization framework to construct classifiers with a desired social burden gap.
}

\section{{Linking selection-rate fairness  and  social burden gap in one-dimensional setting}} \label{sec:one_dim}

We first look at the case when the features are one-dimensional and positive, i.e., $\mathcal{X} = \R_{\geq 0}$.
This setting models several real-world scenarios such as the use of credit scores for loan applications or exam scores for school admissions. 
Furthermore, when the likelihood of positive classification ($\P[Y{=}1 \mid X{=}x]$) can be computed (even approximately), one can use the likelihood as the feature for classification (similar to the model of \citet{milli2019social}).
Secondly, we will assume outcome monotonicity of the cost function with respect to the feature:
if $\tau{>}x_1{>}x_2$, then $c(x_2, \tau){>}c(x_1, \tau)$. 
In this case, threshold-based classifiers will 
classify all individuals with feature values greater than a specific threshold as positive and all individuals with feature values less than the threshold as negative.
As mentioned before, we study cost functions that only have non-zero costs for the individuals classified as negative. Hence, we assume that the cost function $c$ has the following property: $c(x_1, x_2)$ is non-zero (and positive) only when $x_1{<}x_2$; i.e., for a continuous and differentiable function $d:  \mathcal{X}{\times}\mathcal{X}{\rightarrow}\R$, we can say that $c(x_1, x_2) = d(x_1, x_2) \cdot \mathbf{1}(x_2{>}x_1)$.
Due to outcome monotonicity, the gradient of $c(x_1, x_2)$ with respect to $x_1$ will be negative.
We note that these assumptions are similar to those considered in \cite{milli2019social, hu2019disparate}.

Prior work has shown that two kinds of biases can lead to manipulation cost disparities: \textit{feature biases} and \textit{cost function biases}.
For the first part of the analysis, we focus on \textit{feature biases} and we analyze the impact of \textit{cost function biases} later in this section.
Feature biases refer to settings where disadvantaged group individuals have scores concentrated in sub-spaces that have a lower likelihood of positive classification; e.g., credit score datasets exhibit these biases for African-Americans \cite{Hardt2016EqualityOO}.
They can be formally defined as follows.
For a sub-population condition $\psi$, there is feature bias against group $Z=0$ in distribution $\mathcal{D}$ if, for all $x \in \mathcal{X}, z \in \set{0,1}$ s.t. $\Pr[X{<}x \mid Z=z, \psi(X,Y){=}1] \in (0,1)$, we have that 
\begin{equation} \label{def:feature_bias}
\P_\mathcal{D}[X < x \mid Z = 0, \psi(X,Y) = 1] > \P_\mathcal{D}[X < x \mid Z = 1, \psi(X,Y) = 1].
\end{equation}

\noindent
With respect to feature biases, we restate the result of \citet{milli2019social} below, generalizing it to cases when cost analysis may be limited to a certain sub-population defined by condition $\psi$.

\begin{proposition} \label{lem:positive_social_gap}
Suppose we have a sub-population condition $\psi$ and a cost function $c(\cdot, \cdot)$.
For a classifier $f_\tau$, characterized by a single threshold $\tau$, 
if there is feature bias against group $0$
as defined in \eqref{def:feature_bias}
then for all $\tau \in \mathcal{X}$, $G(f_{\tau}, \psi) > 0$.
\end{proposition}

\noindent
Proposition~\ref{lem:positive_social_gap} states that if there is feature bias against group $0$
then using $f_\tau$ leads to higher expected strategic cost for group $0$ than group $1$.
For the one-dimensional setting, the above result shows that a single threshold-based classifier can be discriminatory.
Classifiers that use group-specific thresholds, on the other hand, can achieve low social burden gaps as we show below.
For $\tau_0, \tau_1{\in}\mathcal{X}$, let $f_{\tau_0, \tau_1}$ denote the classifier that uses threshold $\tau_0, \tau_1$ for group 0, 1 respectively.

\begin{proposition} \label{lem:negative_social_gap}
Suppose we are given a sub-population condition $\psi$
and a cost function $c(x_1, x_2)$.
Say there is feature bias against group $0$ in distribution $\mathcal{D}$ as defined in \eqref{def:feature_bias}.
Then there exist $\tau_0, \tau_1 \in \mathcal{X}^2$ such that 
$G(f_{\tau_0, \tau_1}, \psi) < 0$.
\end{proposition}
\noindent
Proposition~\ref{lem:negative_social_gap} shows that appropriately selected group-specific thresholds can lead to a relatively lower social burden for disadvantaged groups; strategies for efficiently searching for these appropriate thresholds are discussed later in this section.
The proofs of both propositions are presented in Appendix~\ref{sec:proofs}.
We next study whether group-specific classifiers that are fair with respect to selection rate fairness metric $H(f, \psi)$ also have low social burden gap $G(f, \psi)$.

As mentioned earlier, one of the goals of fair classification is to provide equal opportunities for all demographic groups.
By equalizing selection rates across groups, prior work forces the classifier to select more disadvantaged group individuals who otherwise would not be selected in the unconstrained case due to dataset biases. 
However, we show below that achieving fairness w.r.t. selection rate-based metrics (like statistical rate) may not lead to reduced strategic manipulation costs for the disadvantaged group.
\nrev{This is because the average strategic cost incurred by a group depends on the distance between the classifier's threshold for the group and the features of negatively classified individuals from this group; the greater this distance, the greater the strategic cost.
Even when a classifier $f$ with fair w.r.t. selection rate fairness metric $H(f, \cdot)$, it may not be fair w.r.t. social burden gap $G(f, \cdot)$ since the distance between classifier threshold and features of negatively-classified individuals of a minority group can still be large (Section~\ref{sec:synthetic} presents simulations on this point).
The following theorem quantifies this issue and shows that the social burden gap depends not just on selection rate disparity, but also on cost function and classifier properties.
}

\begin{theorem} \label{thm:main}
Suppose we are given group-specific thresholds $\tau_0, \tau_1$ and a sub-population condition $\psi$, and the cost function $c(x_1, x_2) = d(x_1,x_2) \mathbf{1} (x_2 > x_1)$. For a fixed $x_2$, suppose that the gradient of $d$ with respect to $x_1$ at any point in $(0, x_2)$ is in the range $[g_l, g_u]$, for some $g_l \leq g_u \leq 0$.
Let $P_z(\tau) = \P[X \in (0, \tau) \mid Z = z, \psi]$ and $E_{z,\tau} = \E[X \mid X \in [0,\tau], Z = z, \psi] P_z(\tau)$.
Then, we can bound the social burden gap of classifier $f_{\tau_0, \tau_1}$ as follows
$$
G(f_{\tau_0, \tau_1}, \psi) \leq g_u \tau_1 H(f_{\tau_0, \tau_1}, \psi) + (g_u\tau_1 - g_l \tau_0) P_0(\tau_0) - g_u E_{1,\tau_1}  + g_l E_{0,\tau_0}, $$ 
$$
G(f_{\tau_0, \tau_1}, \psi) \geq g_l \tau_1 H(f_{\tau_0, \tau_1}, \psi) + (g_l\tau_1 - g_u \tau_0) P_0(\tau_0) - g_l E_{1,\tau_1}  + g_u E_{0,\tau_0}.$$
\end{theorem}
\noindent
The proof is presented in Appendix~\ref{sec:proofs}.
Note that Theorem~\ref{thm:main} does not assume feature bias (Eq~\eqref{def:feature_bias}) to be explicitly present and can handle generic feature distributions.
To interpret the above theorem, consider the impact of different $H(f_{\tau_0, \tau_1}, \psi)$ values.
For simplicity, suppose $\psi(x,y){=}1$ for all $(x,y)$
(i.e., $H(f_{\tau_0, \tau_1}, \psi)$ is the statistical rate).
When $H(f_{\tau_0, \tau_1}, \psi){=}0$, the classifier has an equal selection rate for group 0 and group 1.
Consider the setting when $d$ is linear, i.e., $d(x_1, x_2) = x_2 - x_1$. 
In this case, the upper and lower bounds are equal and the social burden gap is $(\tau_0 - \tau_1)P_0(\tau_0) - E_{0,\tau_0} + E_{1,\tau_1}$. Since $H(f_{\tau_0, \tau_1}, \psi) = 0$, we can simplify $G(f_{\tau_0, \tau_1}, \psi)$ to be $(\tau_0 - \E[X \mid X \in [0,\tau_0], Z = 0] - \tau_1 + \E[X \mid X \in [0,\tau_1], Z = 1]) P_0(\tau_0)$. 

In this equation, $(\tau_z - \E[X \mid X \in [0,\tau_z], Z = z])$ is the average distance of feature values of negative-classified individuals of group $z$ from the decision boundary.
The difference between these distances for group 0 and group 1 
depends on the choice of $\tau_0, \tau_1$ and group distributions. 
To intuitively understand this dependence, we provide simulations over datasets generated using Gaussian distributions in Section~\ref{sec:synthetic}. The simulations show that as the distribution variance increases, the social burden gap of classifiers, constrained to have $H(f_{\tau_0, \tau_1}, \psi){\approx}0$, can also dramatically increase.
This is why simply constraining $H(f_{\tau_0, \tau_1}, \psi)$ is not sufficient to obtain a classifier with a low social burden gap.
However, when $H(f_{\tau_0, \tau_1}, \psi){<}0$, the difference between the above distances is unlikely to be small since (a) low $H$ implies that $\tau_0$ is higher or similar to $\tau_1$ and (b) due to feature bias the feature values of group 0 are lower than group 1.
Hence, $H(f_{\tau_0, \tau_1}, \psi)$ being greater than or equal to 0 is necessary (but not sufficient) to have low social burden gap.

\vspace{0.5em}
\noindent
\textbf{Extension to group-specific cost functions.}
In many settings, strategic cost disparity can arise due to \textit{cost function biases}, i.e., from different groups having different cost functions.
Due to these biases, for the same unit of a feature update, the disadvantaged group would pay a larger cost than the advantaged group; e.g., African Americans face larger access barriers to credit than White Americans \cite{cohen2011credit}.
To account for group-specific costs,
Theorem~\ref{thm:main} can be extended to use group-specific gradient bounds for the cost function; incorporating them leads to the following bounds.
\begin{theorem} \label{thm:main_diff_costs}
Suppose we are given group-specific thresholds $\tau_0, \tau_1$ and a sub-population condition $\psi$. Let $c_z(x_1, x_2) = d_z(x_1,x_2) \mathbf{1} (x_2 > x_1)$ denote the cost for group $z$ individuals. For a fixed $x_2$, suppose that the gradient of $d_z$ with respect to $x_1$ at any point in $(0, x_2)$ is in the range $[g_{l,z}, g_{u,z}]$, for some $g_{l,z} \leq g_{u,z} \leq 0$ for all $z \in \set{0,1}$.
Let $P_z(\tau) := \P[X \in (0, \tau) \mid Z = z, \psi]$ and $E_{z,\tau} = \E[X \mid X \in [0,\tau], Z = z, \psi] P_z(\tau)$.
Then, we can bound the social burden gap of $f_{\tau_0, \tau_1}, \psi)$ as follows
$$
 G(f_{\tau_0, \tau_1}, \psi) \leq g_{u,1} \tau_1 H(f_{\tau_0, \tau_1}, \psi) + (g_{u,1}\tau_1 - g_{l,0} \tau_0) P_0(\tau_0) - g_{u,1} E_{1,\tau_1}  + g_{l,0} E_{0,\tau_0},$$ 
and lower bounded by
$$
G(f_{\tau_0, \tau_1}, \psi) \geq g_{l,1} \tau_1 H(f_{\tau_0, \tau_1}, \psi) + (g_{l,1}\tau_1 - g_{u,0} \tau_0) P_0(\tau_0) - g_{l,1} E_{1,\tau_1}  + g_{u,0} E_{0,\tau_0}.$$
\end{theorem}

\noindent
The proof of the theorem is presented in Appendix~\ref{sec:proofs}.
Cost function biases can be inherently captured using the gradient of the group-specific cost functions. This is because these biases imply that the rate of increase in cost for the disadvantaged group is greater than that of the advantaged group.
Hence, if group 0 faces higher manipulation costs to move to the same point than group 1, then the absolute gradient of $c_0$ will be larger than the absolute gradient of $c_1$; this can result in higher social burden gaps than the case when cost function is same for both groups. 
For instance, suppose $d_0 (x_1, x_2){=}a(x_2{-}x_1)$ and  $d_1 (x_1, x_2){=}(x_2{-}x_1)$ and $H(f, \psi){=}0$.
If $a > 1$, then the social burden gap can be shown to be $(a(\tau_0 - \E[X \mid X \in [0,\tau_0], Z = 0]) - (\tau_1 - \E[X \mid X \in [0,\tau_1], Z = 1])) P_0(\tau_0)$ which is larger than the corresponding social burden gap when cost function is same for both groups.
The above theorem can account for both cost function and feature biases and provides a succinct relationship between standard fairness metrics and social burden gap.

\vspace{0.5em}
\noindent
\textbf{Fair classification with low social burden gap and low selection rate disparity}. 
To construct a classifier that has high institution utility and low social burden gap, we can employ Theorem~\ref{thm:main}, \ref{thm:main_diff_costs}.
Recall that loss function $\mathcal{L}$ measures the expected risk of any classifier.
Standard fair classification algorithms 
already optimize error rate of $f$ w.r.t. $\mathcal{L}$ and subject to selection rate constraints (i.e., $\E_{\mathcal{D}}[\mathcal{L}(f; X,Y)]$) subject to constraints on $H(f, \psi)$. 

To construct low social burden gaps, we can 
alternately use a modified constraint on the upper bound in Theorem~\ref{thm:main} or Theorem~\ref{thm:main_diff_costs} to obtain a classifier with low social burden gap. For example, in the setting of just feature bias, we can use the upper bound from Theorem~\ref{thm:main} as a constraint; i.e., {minimize $\E_{\mathcal{D}}[\mathcal{L}(f; X,Y)]$ subject to  $(g_u\tau_1 - g_l \tau_0) P_0(\tau_0) - g_u E_{1,\tau_1}  + g_l E_{0,\tau_0} \leq 0$}.
Quantities $P_z(\tau_z)$ and $E_{z,\tau_z}$ can be computed empirically for a given dataset:
$P_z(\tau_z), E_{z,\tau_z}$ are the fraction and empirical mean, respectively, of group $z$ individuals who are negatively classified.
In Section~\ref{sec:empirical}, we empirically show that using this modified optimization program results in classifiers that low social burden gaps.

\section{Extension to multiple dimensions} \label{sec:multi_dim}
Suppose that features are $n$-dimensional, for $n{>}1$.
For this multi-dimensional setting, we consider classifiers that 
threshold over a linear combination of the features of the individuals.
We again assume that all features are outcome monotonic, i.e., increasing each feature value results in increase in likelihood of positive classification. Hence, we only consider manipulations from $x_1$ to $x_2$ when $x_2 \geq x_1$, i.e., for all $i \in [n]$, $x_2^{(i)} \geq x_1^{(i)}$.
Finally, the cost function is assumed to be linear, i.e, for an individual from group $z \in \set{0,1}$, the cost function is $c_z(x_1, x_2) = d_z^\top (x_2 - x_1)$ if $x_2 \geq x_1$ and $c(x_1, x_2) =0$ if $x_2 \leq x_1$, given $d_0, d_1 \in \R^n$
(we study the quadratic cost function setting in Appendix~\ref{sec:proofs}).
Linear cost functions have been used in prior work to approximately model real-world strategic settings \cite{hardt2016strategic, hu2019disparate,estornell2021unfairness}.
Vector $d_z$ can encode the different costs paid for updating different features; e.g., in the credit score setting, opening new credit lines has lower costs than increasing annual income.
In this multi-dimensional setting, we can prove the following result.

\begin{theorem} \label{thm:linear_multi}
Suppose we have a linear classifier $f$ such that for an individual with $x$ and group $z$, $f(x) = 1$ if and only if $u^\top x \geq v_z$ and 0 otherwise. For an individual from group $z$ with unmanipulated datapoint $x_1$, the cost to move to point $x_2$ is defined as $c_z(x_1, x_2) = d_z^\top (x_2 -x_1)$ if $x_2 \geq x_1$ and 0 otherwise, for $d_0, d_1 \in \R^n$.
Let 
$w_z^\star := \max_{i \in [n]} u_i/d_{z,i}$.
Then,
$$G(f, \psi) = -\frac{1}{w_1^\star} \left(v_1 H(f, \psi) \right){-}\delta ,$$
where $\delta = \left( \frac{v_1}{w_1^\star} - \frac{v_0}{w_0^\star}\right) P_0 - \frac{1}{w_1^\star} E_{1,v_1}+\frac{1}{w_0^\star} E_{0,v_0} $, $P_z =\P[f(X)=0 \mid Z =z, \psi]$, $E_{z,\tau}=\E[(u^\top X)\mid f(X)=0, Z=z, \psi]P_0$.
\end{theorem}

\noindent
We obtain an equality relation here
since the cost function is linear.
The proof (presented in Appendix~\ref{sec:proofs}) follows by reducing this case to the single-dimensional setting.
This is possible since in the case of linear classifiers the distance of negatively-classified individuals from the classifier threshold can be captured using $u^\top x$.
While limiting the analysis to linear classifiers might seem restrictive, Theorem~\ref{thm:linear_multi} still provides evidence that social burden gap $G(\cdot)$ can be large for classifiers in multi-dimensional settings even when selection rate disparity $H(\cdot)$ is small, due to additional factors captured by $\delta$.
Furthermore, empirical analysis of real-world fairness benchmark datasets (Section~\ref{sec:empirical}) shows that linear classifiers achieve close to state-of-the-art performance for these datasets. Thus, it is important to study constraints on linear classifiers that can ensure low manipulation costs for all groups.

\section{Empirical analysis} \label{sec:empirical}

\begin{figure*}[t]
    \centering
    \small
    \includegraphics[width=\linewidth]{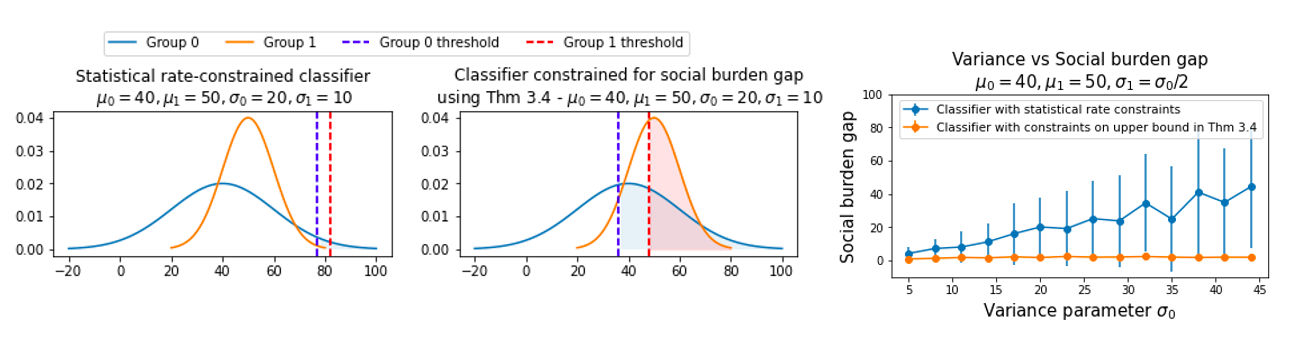}
    \subfloat[Stat. rate constrained classifier]{\hspace{.33\linewidth}}
    \subfloat[Thm 3.4 constrained classifier]{\hspace{.33\linewidth}}
    \subfloat[Variance vs social burden gap]{\hspace{.33\linewidth}}
    \caption{\small Performance of statistical rate constrained classifier and classifier constrained using our method on a synthetic dataset.
    Plots (a), (b) show the distribution and classifier thresholds for one random iteration.
    Plot (c) shows the mean and deviation (over 50 repetitions) of social burden gap of classifiers that are statistical rate-constrained  and classifiers that are constrained using Thm.~\ref{thm:main_diff_costs}.
    }
    \label{fig:syn_results}
\end{figure*}

\subsection{Synthetic Simulation} \label{sec:synthetic}

To intuitively explain different components of Theorem~\ref{thm:main}, \ref{thm:main_diff_costs}, we design a simulation using a synthetic data generation process.
Suppose that features of group $z \in \set{0,1}$ are sampled from $N(\mu_z, \sigma_z)$, where $\mu_0, \mu_1, \sigma_0 \in \R_{>0}$ and $\sigma_1 = \sigma_0/2$. Let $X_z$ denote the features of group $z$. Suppose that 
for element $x_i$ from group $z$, class label $y_i$ is 1 with probability $(x_i + \min(X_z))/(\max(X_z) + \min(X_z))$. We sample 500 elements for each group.
When $\mu_0 < \mu_1$, there will likely be feature bias in this dataset and any classifier $f$ trained over this dataset will have to use group-specific thresholds to achieve statistical parity. Suppose the cost function for manipulation is linear. 
Consider two classifiers trained using the above data.
The first classifier is trained to achieve maximum accuracy subject to $|H(f, \psi)| \leq 0.4$; here $\psi$ is the identity function.
The second classifier is trained to achieve maximum accuracy subject to the constraint that $|G(f, \psi)| \leq 4$.
In other words, the first classifier uses constraints on the statistical rate while the second classifier uses constraints on the social burden gap.
From Figure~\ref{fig:syn_results}a,b, we can see that these classifiers can have different group thresholds.

To understand Theorem~\ref{thm:main}, \ref{thm:main_diff_costs} using this example, note that the bound in both theorems depend on the difference between group-specific quantities $(\tau_z{-}\E[X{\mid}X{\in}[0,\tau_z], Z{=}z])$: the average distance of feature values of negative-classified individuals from group $z$ from the decision boundary.
This difference will increase as $\sigma_0$ increases since within-group variances will increase and the group 0 variance grows faster than group 1 variance (since $\sigma_1 = \sigma_0/2$).
Hence, as $\sigma_0$, increases, the social burden gap of a classifier can increase even when the statistical rate remains small.
We empirically observe this phenomenon in Figure~\ref{fig:syn_results}c, where we see that increasing $\sigma_0$ leads to an increase in the social burden gap of the statistical rate-constrained classifier.
However, we also observe that classifiers constrained using Theorem~\ref{thm:main_diff_costs} have almost-zero social burden gap, demonstrating that using Theorem~\ref{thm:main_diff_costs} for forming fairness constraints leads to classifiers with low social burden gaps for all $\sigma_0$ values.

\subsection{FICO Credit Dataset} \label{sec:fico_main}

\textbf{Dataset.} We use the FICO credit data \cite{Hardt2016EqualityOO} for preliminary real-world data analysis of classifiers that are fair with respect to standard fairness metrics and classifiers that are fair with respect to the social burden gap. 
This dataset contains 116k credit scores corresponding to White individuals and 16k credit scores corresponding to Black/African-American individuals and
a binary class label for loan default for each individual (pre-processing details are provided in Appendix~\ref{sec:emp_appendix}).
As shown by prior work \cite{milli2019social}, this dataset exhibits feature bias against African-American individuals.

\noindent
\textbf{Methodology.} 
Around 20k random samples from the dataset are removed to create a test partition.
Each classifier is composed of two thresholds $(\tau_0, \tau_1)$. Threshold $\tau_0$ is for credit scores of African-American individuals and threshold $\tau_1$ is for credit scores of White individuals.
A classifier assigns a positive class label to an individual if the individual's credit score is larger than the classifier's threshold for the individual's group.
Since the credit scores lie in the range from 1 to 100,  we evaluate all possible classifiers, with $\tau_0, \tau_1$ in the set $\set{1, 2, \dots, 100} \times \set{1, 2, \dots, 100}$, and record their properties.
We use 
the linear cost function $c(x, x') := \mathbf{1}(x > x') \cdot (x - x')$ for this section and 
provide results for the quadratic separable cost function
in Appendix~\ref{sec:fico_appendix}. 
We analyze classifier performance for two sub-population conditions: (a) $\psi_{sr}$ which is always 1, i.e., $\psi_{sr}(x,y){=}1$, for all $x,y$, and (b) $\psi_{tpr}$ which is 1 if true class label is 1, i.e., $\psi_{sr}(x,y){=}\mathbf{1}(y{=}1)$.
$H(f,\psi_{sr})$ measures statistical rate and $H(f,\psi_{tpr})$ measures true positive rate disparity.

\begin{figure*}[t]
    \centering
    \includegraphics[width=\linewidth]{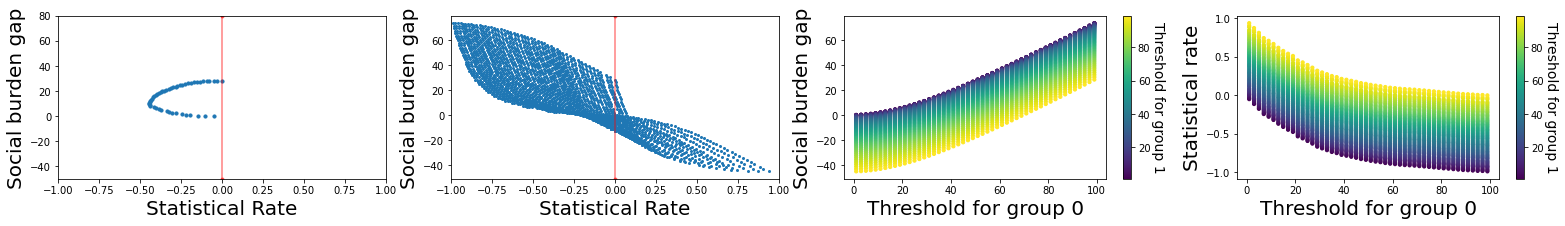}  
    \subfloat[$\tau_0{=}\tau_1{\in}\set{1, \dots, 100}$]{\hspace{.25\linewidth}}
    \subfloat[$(\tau_0, \tau_1){\in}\set{1, \dots, 100}^2$]{\hspace{.25\linewidth}}
    \subfloat[Social burden gap vs $\tau_0$]{\hspace{.25\linewidth}}
    \subfloat[Statistical rate vs $\tau_0$]{\hspace{.25\linewidth}}
    \caption{\small Statistical rate and social burden gap of all classifiers for FICO dataset. Each point represents a classifier and the axes plot different properties of these classifiers.
    Plot (a) presents the social burden gap $G(f, \psi_{sr})$ vs statistical rate $H(f, \psi_{sr})$ for classifiers that use the same threshold for both groups. 
     Plots (b), (c), (d) present social burden gap $G(f, \psi_{sr})$ vs statistical rate $H(f, \psi_{sr})$ for classifiers that can use group-specific threshold. 
     The range of statistical rate and social burden gap values achieved for these classifiers is larger.
    }
    \label{fig:fico_sr_results}
\end{figure*}

\begin{figure*}
    \centering
    \small
    \includegraphics[width=\linewidth]{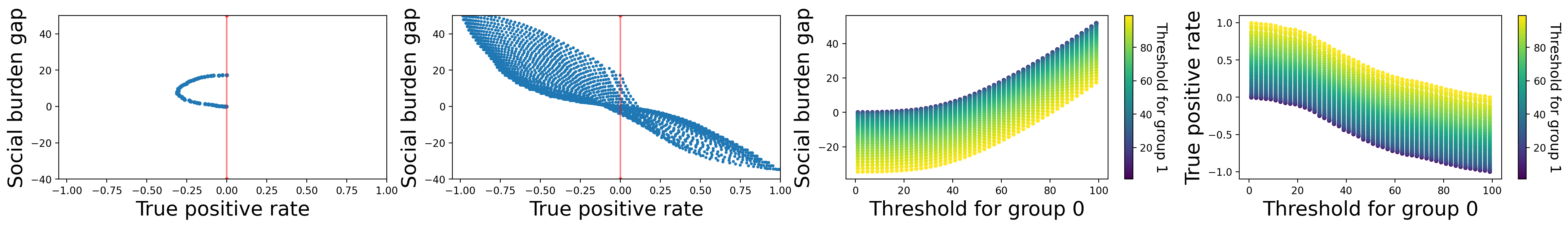}
    \subfloat[$\tau_0{=}\tau_1{\in}\set{1, \dots, 100}$]{\hspace{.25\linewidth}}
    \subfloat[$(\tau_0, \tau_1){\in}\set{1, \dots, 100}^2$]{\hspace{.25\linewidth}}
    \subfloat[Social burden gap vs $\tau_0$]{\hspace{.25\linewidth}}
    \subfloat[True positive rate vs $\tau_0$]{\hspace{.25\linewidth}}
    \caption{\small True positive rate and social burden gap of all classifiers for the FICO dataset.
    Plot (a) presents the social burden gap $G(f, \psi_{tpr})$ vs true positive rate $H(f, \psi_{tpr})$ for single-threshold classifiers. Once again, the true positive rate and social burden gap for these classifiers favors the majority group.
     Plot (b), (c), (d) present $G(f, \psi_{tpr})$ vs $H(f, \psi_{tpr})$ for classifiers that use group-specific thresholds. 
    }
    \label{fig:fico_tpr_results}
\end{figure*}

\noindent
\textbf{Results.}
\textit{Statistical rate $H(\cdot, \psi_{sr})$ vs social burden gap $G(\cdot, \psi_{sr})$.} 
Plot~\ref{fig:fico_sr_results}a presents the results for classifiers that use the same threshold for both groups.
As discussed in Proposition~\ref{lem:positive_social_gap},
these classifiers always have a social burden gap ${\geq}0$ and lead to higher strategic manipulation costs for African-American individuals. Furthermore, even the statistical rate of these classifiers is low implying that all classifiers using single thresholds select White individuals at a higher rate.
Plot~\ref{fig:fico_sr_results}b presents fairness metrics for classifiers that use group-specific thresholds.
Here, we observe that the range of values achieved for statistical rate and low social burden gap is much larger.
A classifier with a high statistical rate favoring the disadvantaged group (${>}0.5$) also has a low social burden gap (${<}{-}10$) for this dataset.
However, almost equal group selection rates do not imply parity with respect to social burden.
For classifiers with statistical rates close to 0,
the social burden gap ranges from $[-12, 30]$.

\begin{figure*}[t]
    \centering
    \small
    \includegraphics[width=\linewidth]{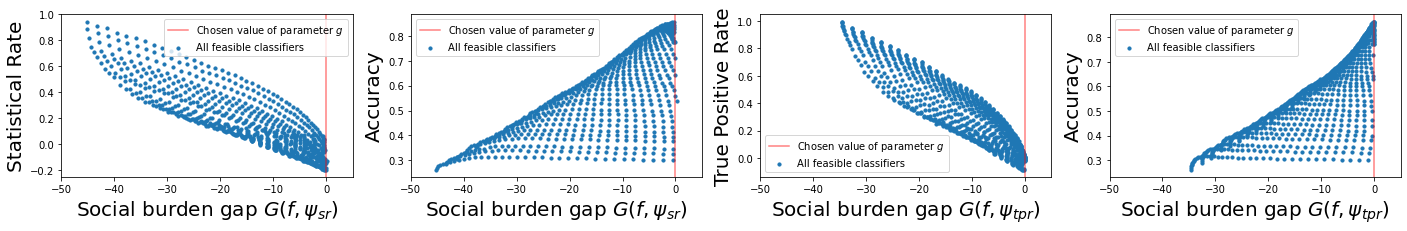}
    \subfloat[Stat. rate vs $G(f, \psi_{sr})$]{\hspace{.25\linewidth}}
    \subfloat[Accuracy vs $G(f, \psi_{sr})$]{\hspace{.25\linewidth}}
    \subfloat[TPR vs $G(f, \psi_{tpr})$]{\hspace{.25\linewidth}}
    \subfloat[Accuracy vs $G(f, \psi_{tpr})$]{\hspace{.25\linewidth}}
    \caption{\small Performance of classifiers that satisfy the modified fairness constraints.
    Plots (a),(b) present the statistical rate and accuracy vs social burden gap $G(f, \psi_{sr})$ for all classifiers for which condition~(\ref{eq}) is satisfied with $\psi = \psi_{sr}$.
    Plots (c),(d) present the true positive rate (TPR) and accuracy vs social burden gap $G(f, \psi_{tpr})$ for all classifiers for which condition~(\ref{eq}) is satisfied with $\psi = \psi_{tpr}$.
    }
    \label{fig:fico_fair_clf}
\end{figure*}

\textit{True positive rate $H(\cdot, \psi_{tpr})$ vs social burden gap $G(\cdot, \psi_{tpr})$.} 
With $\psi_{tpr}$,
we compute costs for individuals who are incorrectly negatively classified.
From Plot~\ref{fig:fico_tpr_results}a, we again see that single-threshold classifiers have social burden gap ${\geq}0$.
Plot~\ref{fig:fico_tpr_results}b, however, shows that there exist group-specific thresholds that result in low social burden gap and high true positive rate for African-American individuals.$\\$
Plots~\ref{fig:fico_sr_results}c, d, \ref{fig:fico_tpr_results}c, d %
present the relationship between group thresholds and fairness metrics.
Increasing group 0 threshold increases the social burden for group 0 and decreases the statistical/true positive rate as positive classifications decrease.

\textit{Constructing classifiers with low social burden gap.} 
We next empirically analyze the inequalities in Theorem~\ref{thm:main}.
Suppose the goal of the institution is to maximize accuracy subject to the constraint that social burden gap is ${\leq}g$, for some $g{\in}\R$. 
Since group 0 is the marginalized one, the institution aims to achieve a non-positive social burden gap to address the manipulation cost disparities.
As shown in \ref{thm:main}, social burden gap is upper bounded by $g_u \tau_1 H(f, \psi) + (g_u\tau_1 - g_l \tau_0) P_0(\tau_0) - g_u E_{1,\tau_1}  + g_l E_{0,\tau_0}$.
For the linear cost function, $g_u = g_l = 1$. Hence, we use the burden gap constraint 
\setlength{\abovedisplayskip}{1.5pt}
\setlength{\belowdisplayskip}{1.5pt}
\begin{equation}
\tau_1 H(f, \psi) - (\tau_1 - \tau_0) P_0(\tau_0) + E_{1,\tau_1} - E_{0,\tau_0}  \leq g. \label{eq}  
\end{equation}
For $g{=}0$, Figure~\ref{fig:fico_fair_clf} plots the accuracy and social burden gap of all classifiers that satisfy the above conditions for $\psi_{sr}$ and $\psi_{tpr}$.
The plots show that all classifiers that satisfy the constraint on the upper bound from Theorem~\ref{thm:main} satisfy the condition $G(f_{\tau_0, \tau_1}, \psi_{sr}){<}0$.
Furthermore, even in this case, the classifier that optimizes this constrained problem has high accuracy.
For $\psi_{sr}$ the accuracy of the optimal constrained classifier is 0.86 and $G(f_{\tau_0, \tau_1}, \psi_{sr}){=}{-}0.21$ and for $\psi_{tpr}$ the accuracy of the optimal constrained classifier is 0.86 and $G(f_{\tau_0, \tau_1}, \psi_{tpr}){=}{-}0.03$.
\nrev{In comparison, the accuracy of the optimal unconstrained classifier is 0.88, showing minimal loss in accuracy due to the constraints.}

\subsection{Adult Income Dataset}

\textbf{Dataset.} For analysis of multi-dimensional data, we use the Adult Income dataset.
We use the new version of this dataset developed and preprocessed by \citet{ding2021retiring}.
It contains information on around 251k individuals from the state of California surveyed in 2019. The classification task is to predict whether the income of an individual is above \$50k or not.
The strategic features available are ``class of worker'', ``occupation'', and ``hours worked per week'' (the other five features are listed in Appendix~\ref{sec:emp_appendix}).
We use race as the protected attribute, limiting the dataset to White (93\% of the dataset; $z=1$) and Black/African-American (7\% of the dataset; $z=0$) individuals.

\noindent
\textbf{Methodology.} 
The cost function used is linear and group-specific.
Let $d'{\in}\R^9$ be the underlying cost vector such that $d'_i{=}100$ if $i$ represents ``class of worker'' feature, $d'_i{=}10$ if $i$ represents ``occupation'', $d'_i{=}1$ if $i$ represents ``hours per week'', $d'_i{=}\infty$ for other non-strategic features.
$d'$ assigns a higher cost factor depending on the difficulty of updating a feature value.
The cost function for group 0 is $c_0(x, x'){:=}2{\cdot}d'^\top (x'{-}x)$ and cost function for group 1 is $c_1(x, x') := d'^\top (x'{-}x)$.
In this case, African-American individuals pay twice the cost that White individuals pay for the same feature update.
The dataset is partitioned into 80-20 random train-test splits.
Once again, suppose $\psi_{sr}(x,y){=}1$ for all $(x,y)$.

We will restrict the classifiers to be from the linear family and
use the logistic log-loss function $\mathcal{L}(f; x,y) := {-}y\log \sigma(f(x)) - (1{-}y)\log (1{-}\sigma(f(x)))$ to measure prediction error of $f$
(here $\sigma(\cdot)$ is the standard sigmoid function). 
We analyze the following different classifiers.
Classifier $f_{uncons}{:=}\arg\min_f \E[L(f; X,Y)]$ will denote the unconstrained classifier.
For pre-specified desired statistical rate $\epsilon{\in}[-1,1]$,
classifier $f_{sr}:= \arg \min_f \E[L(f; X,Y)]$ subject to $H(f, \psi_{sr}){\geq}\epsilon$.
Finally, for a pre-specified $g{\in}\R$, we construct classifiers with social burden gap $G(f, \psi_{sr}) \leq g$.
To do so, we use the result from Section~\ref{sec:multi_dim} and 
construct classifier $f_{strat} := \arg \min_f \E[L(f; X,Y)]$ subject to 
$-\frac{1}{w_1^\star} \left(v_1 H(f, \psi) \right){-}\delta \leq g$ (quantities $w_1^\star, v_1, \delta$ are defined in Theorem~\ref{thm:linear_multi}).
We will set $\epsilon{=}0$ and $g{=}0$ to analyze classifiers with equal selection rate and zero social burden gap in this section, and present variation of performance with these parameters in Appendix~\ref{sec:adult_appendix}.
\cg{To compare with another fair classification baseline, we also implement the fair logistic regression algorithm of \citet{rezaei2020fairness} with statistical rate constraints; we will call this classifier $f_{RFMZ}$. }
We report the mean and standard error of accuracy, statistical rate, and social burden gap of all classifiers over 100 random train-test splits.
Implementation details of all methods are provided in Appendix~\ref{sec:emp_appendix}.
\begin{table}[t]
\caption{\small Performance of the unconstrained classifier $f_{uncons}$, classifier constrained to achieve statistical rate $\geq 0$ $f_{sr}$, fair classifier from Rezaei et al. $f_{RFMZ}$, and classifier constrained to achieve low social burden gap using our method $f_{strat}$ on the Adult dataset.}
\begin{center}
\begin{small}
\begin{sc}
\begin{tabular}{llcccr}
\toprule
&Classifier & Accuracy & Stat. Rate & Burden gap \\
\midrule
\multirow{3}{*}{Baselines}  & $f_{uncons}$ & \textbf{0.84} $\pm$ 0.00 & -0.11 $\pm$ 0.01 & 2.23 $\pm$ 0.04\\
& $f_{sr}$ & 0.83 $\pm$ 0.01 & 0.01 $\pm$ 0.03 &  1.12 $\pm$ 0.15 \\
& $f_{RFMZ}$ & 0.83 $\pm$ 0.00 & -0.08 $\pm$ 0.01 &  1.88 $\pm$ 0.08 \\
\midrule
Our method & $f_{strat}$ & 0.82 $\pm$ 0.00 & \textbf{0.24} $\pm$ 0.01 & \textbf{0.04} $\pm$ 0.01 \\
\bottomrule
\end{tabular}
\end{sc}
\end{small}
\end{center}
\label{tab:adult_results}
\end{table}

\noindent
\textbf{Results.}
Table~\ref{tab:adult_results} presents the performance of classifiers $f_{uncons},$ $f_{sr}, f_{RFMZ}, f_{strat}$.
First note that the unconstrained classifier $f_{uncons}$ has an average statistical rate of -0.11 (i.e., the selection rate of African-American individuals is much lower than the selection rate of White individuals) and the average social burden gap is 2.23 (i.e, cost of strategic manipulation is higher for African-Americans).
Hence, fairness interventions are necessary in this case to achieve equal performance.
For classifier $f_{sr}$, the statistical rate is close to 0; however, the social burden gap is still greater than 0 in this case, showing that an almost equal selection rate does not necessarily imply a low social burden gap.
\cg{Similarly, baseline $f_{RFMZ}$ has a high social burden gap despite having a better statistical rate than the unconstrained classifier.
}
In comparison, our classifier $f_{strat}$ has a statistical rate of 0.24, i.e., the selection rate for African-Americans is higher.
Furthermore, classifier $f_{strat}$ achieves a social burden gap close to 0 on average. Hence, both groups pay almost equal 
cost for strategic manipulation.
In terms of accuracy, $f_{uncons}$ achieves the highest accuracy (0.84) and the accuracy of $f_{strat}$ is only slightly lower (0.82),
implying a minimal loss in accuracy due to fairness constraints.

\section{Discussion and Limitations} \label{sec:limitations}
Our paper untangles the relationship between manipulation cost disparities and standard fairness metrics and provides a framework to construct classifiers with low costs for minority groups.
The proposed framework can be useful for real-world classification settings where individuals repeatedly interact with an institution and its classifier (e.g., applying for loans after rejection with updated features).
In these settings, the institution can employ our framework to ensure that minority individuals do not pay disparately higher costs to exercise their available recourse options.
In this section, we discuss philosophical grounding, practical advantages, and certain limitations
of our framework.

\noindent
\textbf{Philosophical grounding.}
\citet{venkatasubramanian2020philosophical} argued how recourse options 
can be systematically helpful to groups that have been historically oppressed.
In particular, they distinguish between ``\textit{token acts of exercising recourse (reversing a single harmful decision) and the general state of enjoying systematic access to the power to reverse harmful decisions (knowing that if a harmful decision were to be made, one would be able to get it reversed)}.''
Recognizing this distinction at an institutional level 
motivates the construction of classifiers that equalize manipulation costs across groups to ensure systematic access to recourse for all.
This way, the institution can acknowledge biases in data and costs and provide redressal mechanisms that account for structural inequalities.

\noindent
\textbf{Practical advantages and information required by the institution to address manipulation disparities.} Beyond the presented analysis, our framework has advantages that allow for easy use in applications. For instance, the bounds on the social burden gap
require minimal information about cost functions
and can handle settings where cost functions vary across individuals but belong to the same class of gradient-bounded or Lipschitz functions.
Hence, an institution aiming to construct classifiers with a low social burden gap would not need to model the complete update behaviors of every individual.
Nevertheless, some information about the cost function gradient is required for implementing our framework in practice. While this can be quantified by observing the past and current behavior of individuals who have been negatively classified, errors in cost function gradient measurement can affect the performance of our proposed fairness intervention. Future work can additionally explore methods to reduce strategic manipulation disparities while ensuring that the method is robust to gradient measurement errors.

A social burden gap of 0 implies equal manipulation costs for all groups; however, in some cases, a gap of less than 0 may be necessary.
In multi-feedback settings, the predictions at one time step affect the classifier training at the next time step \cite{perdomo2020performative} and biases can be amplified across feedback iterations.
In these cases, our framework can also be used to construct a classifier with a negative social burden gap 
to tackle these 
biases.

\noindent
\textbf{Negative manipulations. }
Strategic manipulations can potentially be used to ``game'' a classifier \cite{hardt2016strategic} and, by reducing manipulation costs, our methods can potentially lead to increased ``gamification''.
The issue of gamification primarily arises due to noisy features that are only superficially predictive of the class label. In the absence of other robust features, a classifier will use these noisy features for prediction.
However, the presence of noisy features does not warrant disparity in manipulation costs, especially if these features favor the majority group.
Nevertheless, recent papers have also suggested classifier designs that incentivize \textit{positive manipulations} to improve individuals' task-related qualities
\cite{kleinberg2020classifiers}; using our framework with these classifiers can ensure that all groups have equal opportunities for improvement.

\noindent
\textbf{Model limitations.}
For multi-dimensional settings, our framework considers linear and quadratic cost functions.
Extensions of our framework for generic cost functions in multi-dimensional settings can be further studied as part of future work.
Secondly, while we consider binary protected attributes in our analysis, our results can be extended to non-binary attributes.
This is because the non-binary setting can be reduced to the binary setting by considering the pairwise comparison of measures for different protected attribute values.
However, due to multiple comparisons, the bounds for $G(\cdot,\cdot)$ will be weaker, and future work can explore ways to improve these bounds for non-binary attributes. 

Additional limitations of our framework are related to the accuracy of information about the individuals available to the institution. 
As mentioned earlier, if the institution does not have accurate information about the costs 
associated with feature updates, then our proposed framework might not be completely effective in addressing manipulation disparities.
Another information-based limitation
is the assumption that the classifier used by the institution is known. 
This may not be true in real-world settings and only partial information about decision rules
may be publicly available.
Recent work by \citet{ghalme2021strategic} aims to address this problem in general strategic classification settings and our framework can potentially be extended in the future along similar lines.

\section{Conclusion}

We study the impact of fair classifiers
on individuals' ability to positively manipulate their features based on their group membership.
In settings where feature distributions or cost functions are biased against minority groups, we observe that classifiers can have (almost) equal selection rates for all groups but can still have relatively higher costs for strategic manipulation for minority groups.
We propose modified fairness constraints to construct classifiers that reduce this 
disparity and show its efficacy over multiple datasets.
Our work demonstrates the necessity of analyzing the impact of fair classifiers in dynamic settings and developing approaches that provide equal opportunities independent of group memberships.

\clearpage

\bibliography{references}
\bibliographystyle{plainnat}

\clearpage

\appendix
\section{Proofs} \label{sec:proofs}

\textbf{Proof of Proposition~\ref{lem:positive_social_gap}.}
Let $p_z$ denote $\P[X \mid \psi(X,Y)=1, Z=z]$.
Note that the threshold $\tau$ is the same for all individuals here.
The cost assigned to group $Z=z$ by this classifier can be quantified as the distance of $x$ to threshold $\tau$ due to the monotonicity assumption and since this move maximizes the individual's utility.
\[B_z = \E[c(x, z) \mid \psi, Z=z] = \int_{\mathcal{X}} c(x, \tau) p_z(x) dx = \int_{0}^{\tau} d(x, \tau) p_z(x) dx.\]
The second equality follows from the fact that individuals with $f(x) = 1$ pay zero cost.
Performing integration by parts, we get
\[B_z = \int_{0}^{\tau} d(x, \tau) p_z(x)dx = -\int_{0}^{\tau} \pdv{d(x, \tau)}{x} P_z(x) dx,\]
where $P(x) = \int_{0}^x p(x) dx$ is the cumulative density function.
Therefore, the disparity will be 
\begin{align*}
G(f_{\tau}, \psi) &= B_0 - B_1
= \int_{0}^{\tau} \pdv{d(x, \tau)}{x} P_1(x) dx - \int_{0}^{\tau} \pdv{d(x, \tau)}{x} P_0(x) dx\\ &= \int_{0}^{\tau} \pdv{d(x, \tau)}{x} \left( P_1(x) - P_0(x) \right) dx.
\end{align*}

\noindent
Note that $\pdv{d(x,\tau)}{x}$ is negative due to the monotonicity assumption and $P_1(x) < P_0(x)$ by definition of feature bias. Hence, we get the result of \citet{milli2019social} that $G(f_{\tau}, \psi) > 0$.

\vspace{0.5em}
\noindent
\textbf{Proof of Proposition~\ref{lem:negative_social_gap}.}
Let $p_z$ denote $\P[X \mid \psi(X,Y)=1, Z=z]$.
For the strategic cost function $c$, suppose that the gradient of the function $c$ with respect to $x_1$ at any point is in the range $[g_l, g_u]$, for some $g_l \leq g_u \leq 0$.
From to the previous proof, we know that the cost assigned to group $Z=z$ by this classifier is 
\[B_z = \int_{0}^{\tau_z} d(x, \tau_z) p_z(x)dx = -\int_{0}^{\tau_z} \pdv{d(x, \tau_z)}{x} P_z(x) dx,\]
Therefore, the disparity will be $G(f_{\tau_0, \tau_1}, \psi) = B_0 - B_1$
\begin{align*}
 = \int_{0}^{\tau_1} \pdv{d(x, \tau_1)}{x} P_1(x) dx - \int_{0}^{\tau_0} \pdv{d(x, \tau_0)}{x} P_0(x) dx
\end{align*}
When $\tau_0 = \tau_1$, we get positive social burden gap from Proposition~\ref{lem:positive_social_gap}.
When $\tau_0 < \tau_1$, the expression can take different values depending on choice of these thresholds.
\begin{align*}
G(f_{\tau_0, \tau_1}, \psi) = B_0 - B_1  
& \leq g_u \int_{0}^{\tau_1} P_1(x) dx -  g_l \int_{0}^{\tau_0} P_0(x) dx \\
& \leq g_u \int_{0}^{\tau_1} P_0(x) dx -  g_l \int_{0}^{\tau_0} P_0(x) dx,
\end{align*}
since $P_1(x) \leq P_0(x)$. Using change of variables, $x' = x \cdot g_u/g_l $,
\begin{align*}
&G(f_{\tau_0, \tau_1}, \psi) \leq g_u \int_{0}^{\tau_1} P_0(x) dx -  g_l \int_{0}^{\tau_0} P_0(x) dx\\
&= g_u \int_{0}^{\tau_1} P_0(x) dx -  g_u \int_{0}^{ \tau_0 g_u/g_l} P_0(x') dx' = g_u \int_{ \tau_0 g_u/g_l}^{\tau_1} P_0(x) dx.
\end{align*}
Since, $g_u < 0$, the above term is non-positive if $\tau_0 \leq \frac{g_l}{g_u} \tau_1.$

\vspace{1em}
\noindent
\textbf{Proof of Theorem~\ref{thm:main}.}
Let $p_z := \P[X \mid \psi(X,Y)=1, Z=z]$. By definition,
\begin{align*}
    &G(f_{\tau_0, \tau_1}, \psi) = \E[c(X, \tau_0) \mid \psi, Z=0]  -  \E[c(X, \tau_1) \mid  \psi,Z=1]\\
    &= \int_{\mathcal{X}} d(x,\tau_0) \mathbf{1}(x < \tau_0) p_0(x) dx - \int_{\mathcal{X}} d(x, \tau_1) \mathbf{1}(x < \tau_1) p_1(x)dx\\
    &= \int_{0}^{\tau_0} d(x,\tau_0) p_0(x) dx - \int_{0}^{\tau_1} d(x,\tau_1) p_1(x) dx.
\end{align*}

\noindent
Analyzing each term above individually and integrating by parts,
\begin{align*}
  \int_{0}^{\tau_z} d(x,\tau_z) p_z(x) dx = d(x,\tau_z) P_z(x)|_{0}^{\tau_z} - \int_{0}^{\tau_z} \pdv{d(x,\tau_z)}{x} P_z(x) dx.
\end{align*}
where $P_z(x) = \int_{0}^x p(x) dx$ is the cumulative density function.
Since $P_z(0) = 0$ and $d(\tau_z, \tau_z) = 0$, we have that
\begin{align*}
   G(f_{\tau_0, \tau_1}, \psi) = \int_{0}^{\tau_1} \pdv{d(x,\tau_1)}{x} P_1(x) dx - \int_{0}^{\tau_0} \pdv{d(x,\tau_0)}{x} P_0(x) dx.
\end{align*}
Next, using gradient bounds, we can simplify each term above:
\[g_l \int_{0}^{\tau_z} P_z(x) dx \leq \int_{0}^{\tau_z} \pdv{d(x,\tau_z)}{x} P_z(x) dx \leq g_u \int_{0}^{\tau_z} P_z(x) dx.\]
Once again, integrating by parts, we get that $\int_{0}^{\tau_z} P_z(x) dx =$
\begin{align*}
 & [x P_z(x) - \E[X \mid X \in [0,x], Z=z, \psi] \P[X \in [0,x], Z=z, \psi]]_0^{\tau_z}\\ 
&= \tau_z P_z(\tau_z) - E_{z, \tau_z} \P[X \in [0,\tau_z], Z=z, \psi].
\end{align*}
where $E_{z, \tau_z} := \E[X \mid X \in [0,\tau_z], Z=z, \psi] P_z(\tau_z)$.
We can similarly write selection rate disparity $H(\cdot)$ as
\begin{align*}
H(f_{\tau_0, \tau_1}, \psi) 
&= \int_{\tau_0}^{\infty} p_0(x) dx - \int_{\tau_1}^{\infty} p_1(x) dx \\
& = \int_{0}^{\tau_1} p_1(x) dx -\int_{0}^{\tau_0} p_0(x) dx = P_1(\tau_1) - P_0(\tau_0).
\end{align*}
Now substituting the expressions for $\int_{0}^{\tau_z} P_z(x) dx$ and $H(\cdot)$ in the bounds for $G(\cdot)$, we get that
\begin{align*}
   &G(f_{\tau_0, \tau_1}, \psi) \leq g_u \tau_1 P_1(\tau_1) - g_u E_{1,\tau_1} - g_l  \tau_0 P_0(\tau_0) + g_l E_{0,\tau_0} \\
   &= g_u \tau_1 P_1(\tau_1) - g_l  \tau_0 P_0(\tau_0) - g_u E_{1,\tau_1}  + g_l E_{0,\tau_0} \\
   &= g_u \tau_1 H(f_{\tau_0, \tau_1}, \psi) + g_u \tau_1 P_0(\tau_0) - g_l  \tau_0 P_0(\tau_0) - g_u E_{1,\tau_1}  + g_l E_{0,\tau_0} \\
   &\leq g_u \tau_1 H(f_{\tau_0, \tau_1}, \psi) + (g_u\tau_1 - g_l \tau_0) P_0(\tau_0) - g_u E_{1,\tau_1}  + g_l E_{0,\tau_0}.
\end{align*}
Similarly, for the lower bound
\begin{align*}
   &G(f_{\tau_0, \tau_1}, \psi) \geq g_l \tau_1 P_1(\tau_1) - g_l E_{1,\tau_1} - g_u  \tau_0 P_0(\tau_0) + g_u E_{0,\tau_0} \\
   &= g_l \tau_1 H(f_{\tau_0, \tau_1}, \psi) + g_l \tau_1 P_0(\tau_0) - g_u  \tau_0 P_0(\tau_0) - g_l E_{1,\tau_1}  + g_u E_{0,\tau_0} \\
  &\geq g_l \tau_1 H(f_{\tau_0, \tau_1}, \psi) + (g_l\tau_1 - g_u \tau_0) P_0(\tau_0) - g_l E_{1,\tau_1}  + g_u E_{0,\tau_0}.
\end{align*}

\vspace{0.5em}
\noindent
\textbf{Proof of Theorem~\ref{thm:main_diff_costs}.}
Let $p_z$ denote $\P[X \mid \psi(X,Y)=1, Z=z]$.
\begin{align*}
    &G(f_{\tau_0, \tau_1}, \psi) = \E[c_0(X, \tau_0) \mid \psi, Z=0]  -  \E[c_1(X, \tau_1) \mid  \psi,Z=1]\\
    &= \int_{\mathcal{X}} c_0(x,\tau_0) p_0(x) dx - \int_{\mathcal{X}} c_1(x,\tau_1) p_1(x)dx\\
    &= \int_{\mathcal{X}} d_0(x,\tau_0) \mathbf{1}(x < \tau_0) p_0(x) dx - \int_{\mathcal{X}} d_1(x, \tau_1) \mathbf{1}(x < \tau_1) p_1(x)dx\\
    &= \int_{0}^{\tau_0} d_0(x,\tau_0) p_0(x) dx - \int_{0}^{\tau_1} d_1(x,\tau_1) p_1(x) dx.
\end{align*}

\noindent
From the proof of Theorem~\ref{thm:main}, we know that
\begin{align*}
   G(f_{\tau_0, \tau_1}, \psi) = \int_{0}^{\tau_1} \pdv{d_1(x,\tau_1)}{x} P_1(x) dx - \int_{0}^{\tau_0} \pdv{d_0(x,\tau_0)}{x} P_0(x) dx.
\end{align*}
Using the inequalities from the proof of Theorem~\ref{thm:main}, we get the following upper bound.
\begin{align*}
   &G(f_{\tau_0, \tau_1}, \psi) \leq g_{u,1} \tau_1 P_1(\tau_1) - g_{u,1} E_{1,\tau_1} - g_{l,0}  \tau_0 P_0(\tau_0) + g_{l,0} E_{0,\tau_0} \\
   &\leq g_{u,1} \tau_1 H(f_{\tau_0, \tau_1}, \psi) + (g_{u,1}\tau_1 - g_{l,0} \tau_0) P_0(\tau_0) - g_{u,1} E_{1,\tau_1}  + g_{l,0} E_{0,\tau_0}.
\end{align*}
Similarly, for the lower bound
\begin{align*}
   &G(f_{\tau_0, \tau_1}, \psi) \geq g_{l,1} \tau_1 P_1(\tau_1) - g_{l,1} E_{1,\tau_1} - g_{u,0}  \tau_0 P_0(\tau_0) + g_u E_{0,\tau_0} \\
  &\geq g_{l,1} \tau_1 H(f_{\tau_0, \tau_1}, \psi) + (g_{l,1}\tau_1 - g_{u,0} \tau_0) P_0(\tau_0) - g_{l,1} E_{1,\tau_1}  + g_{u,0} E_{0,\tau_0}.
\end{align*}

\vspace{1em}
\noindent
\textbf{Proof of Theorem~\ref{thm:linear_multi}.}
We will use the result of \citet{hu2019disparate} for this analysis. In particular, they prove the following theorem.

\begin{theorem}[\cite{hu2019disparate}] \label{thm:hu_result}
Suppose we have a linear classifier $f_{u,v}$ such that $f(u,v)=1$ if $u^\top x \geq v$ and 0 otherwise. Consider a candidate with unmanipulated feature $x$ and linear costs $c(x, x') = d^\top (x'-x)$ to move to feature $x' \geq x$ for a given $d{\in}\R^n$. Let $K = \arg \max_{i \in [n]} u_i/d_i$. Then the candidate's best strategy for manipulation is to move to point
$x' := x + \sum_{i=1}^n  e_i t_i/d_i,$
where $e_i$ is the unit vector for component $i$, $t_i \geq 0$ for all $i \in [n]$,  $t_i =0 $ for $i \notin K$, and $u^\top x' = v$.
\end{theorem}
\noindent
Note that we have group-specific classifiers of the form $u^\top x \geq v_0$ for group 0 and $u^\top x \geq v_1$ for group 1.

Let $K_z := \arg \max_{i \in [n]} u_i/d_{z,i}$
Then for an individual with feature $x$ and group $z$ such that $f(x) = 0$, the cost paid by this individual to achieve a positive classification is $\sum_{i\in K_z} t_i$ (from Theorem~\ref{thm:hu_result}).
Also note that $x' = x+ \sum_{i \in K_z} \frac{t_i}{d_{z,i}} e_i$. Then 
\begin{align*}
u^\top (x+  \sum_{i \in K_z} \frac{t_i}{d_{z,i}} e_i) = v_z &\implies u^\top x+ \sum_{i \in K_z} \frac{u_i}{d_{z,i}} t_i = v_z\\ &\implies \sum_{i\in K_z} t_i = \frac{1}{w_z^\star} (v_z - u^\top x),
\end{align*}
where $w_z^\star = \max_{i \in [n]} u_i/d_{z,i}$.
To reduce the single-dimensional setting, suppose the individual's feature is basically $v = u^\top x$.
The cost for update is then $\frac{1}{w_z^\star} (v_z -v)$.
The gradient of this cost function with respect to $v$ is $g_{u,z} = g_{l,z} = -\frac{1}{w_z^\star}$.
Then, using \ref{thm:main_diff_costs} and since the cost function is linear, we can upper and lower bound the social burden gap by
\[\textstyle
-\left( \frac{1}{w_1^\star} v_1 H(f_{u, v_0, v_1}, \psi) + (\frac{1}{w_1^\star} v_1 - \frac{1}{w_0^\star}v_0)P_0 - \frac{1}{w_1^\star} E_{1,v_1}  + \frac{1}{w_0^\star} E_{0,v_0} \right),\]
$E_{z,\tau} = \E[(u^\top X) \mid f(X)=0, Z=z, \psi]\P[f(X)=0 \mid Z=z, \psi]$.
Therefore,
\[G(f_{u, v_0, v_1}, \psi) = -\frac{1}{w_1^\star} \left(v_1 H(f_{u, v_0, v_1}, \psi) \right) - \delta ,\]
where $\delta = \left( \frac{v_1}{w_1^\star} - \frac{v_0}{w_0^\star}\right) P_0 - \frac{1}{w_1^\star} E_{1,v_1}  + \frac{1}{w_0^\star} E_{0,v_0} $

\vspace{1em}
\noindent
\textbf{Result for quadratic functions.}
In the multi-dimensional setting, assume that the cost function is a generalized squared interpoint distance;
i.e., for a given positive definite matrix $B$, the cost of moving from feature $x$ to $x'$ is $c(x,x') = (x'-x)^\top B (x' - x)$ if $x' \geq x$ and 0 otherwise.
In this multi-dimensional setting, we can prove the following result.

\begin{theorem} \label{thm:quad_multi}
Suppose we have a linear classifier $f$ such that for an individual with $x$ and group $z$, $f(x)=1$ if and only if $u^\top x \geq v_z$ and 0 otherwise. For an individual from group $z$ with unmanipulated datapoint $x$, the cost to move to point $x' \geq x$ is defined as $c(x,x') = (x'-x)^\top B (x' - x)$ if $x' \geq x$ and 0 otherwise, for symmetric positive-definite $B \in \R^n \times \R^n$.
Then,
\[\textstyle G(f, \psi) \leq -\frac{2\max_{x} (v_z - u^\top x)}{u^\top B^{-1} u}  (v_0 H(f, \psi) - v_0 P_1 + E_{0}), \]
where, $E_{z} := \E[u^\top X{\mid}\psi, Z{=}z, f(X){=}0]P_z$ and $P_z := \P[f(X){=}0{\mid}\psi, Z{=}z]$.
\end{theorem}

\begin{proof}
Let $p_z$ denote $\P[X \mid \psi(X,Y)=1, Z=z]$.
Suppose we have an individual with feature $x$ such that $f(x) = 0$.
Then the ideal update for the individual $x' := \min_{x': f(x')=1} (x-x')^\top B(x-x')$.

Using the Lagrangian method, for a parameter $\lambda > 0$, the Lagrangian function for this optimization program is
$L(\lambda, x') := (x-x')^\top B(x-x') + \lambda (u^\top x' - v).$
Taking the derivative with respect to $x'$ we get
\[\pdv{L}{x'} = 2 B (x-x') + \lambda u.\]
Setting the derivative to zero, we get
\[x' = x + \frac{1}{2} \lambda B^{-1} u,\]
\[\lambda = \frac{2(v_z - u^\top x)}{u^\top B^{-1} u}.\]
Therefore, the strategic cost incurred by this individual is
\[\frac{1}{4} \lambda^2 u^\top B^{-1} u =\frac{1}{4} \left(\frac{2(v_z - u^\top x)}{u^\top B^{-1} u} \right)^2 u^\top B^{-1} u = \frac{(v_z - u^\top x)^2}{u^\top B^{-1} u}. \]
Using this expression, the cost paid to update $x$ is reduced to a single quantity that uses $u^\top x$ for every individual.
Let $w = u^\top x$. Then with $w$ as feature, we can directly use the one-dimensional case.
In particular, gradient of cost function is upper bounded by $g_{u,z}=0$ and lower bounded by $g_{l,z} = -\frac{2\max_{x} (v_z - u^\top x)}{u^\top B^{-1} u}$.
Let $E_{z} := \E[u^\top X \mid \psi, Z =z, f(X) = 0]P_z$ and $P_z := \P[f(X) = 0 \mid \psi, Z =z]$
Therefore using Theorem 3.5, 
\begin{align*}
    G(f, \psi) &\leq g_{u,1} v_1 H(f, \psi) + (g_{u,1}v_1 - g_{l,0} v_0) P_0 - g_{u,1} E_{1}  + g_{l,0} E_{0} \\
    &= - g_{l,0} v_0 P_0  + g_{l,0} E_{0} \\
    &= g_{l,0} v_0 (H(f, \psi) - P_1)   + g_{l,0} E_{0} \\
    &= -\frac{2\max_{x} (v_z - u^\top x)}{u^\top B^{-1} u}  (v_0 H(f, \psi) - v_0 P_1 + E_{0}).
\end{align*}

\end{proof}

\section{Additional empirical details} \label{sec:emp_appendix}

\textbf{FICO dataset.} The dataset was initially built from the analysis presented in the Report to the Congress on Credit Scoring and Its Effects on the Availability and Affordability of Credit by the Federal Reserve \cite{barocas2017fairness}. 
In particular, the dataset contains FICO scores of TransUnion, a US-based credit scoring agency, provided to the Federal Reserve on request.
This dataset \cite{Hardt2016EqualityOO} \footnote{\url{github.com/fairmlbook/fairmlbook.github.io/tree/master/code/creditscore/data}}
contains cumulative density functions for credit scores of African-American and Caucasian individuals as well as the likelihood of \textit{defaulting} on a loan conditional on race and credit score.
Using the cumulative density functions, we sample around 116k credit scores corresponding to white individuals and 16k credit scores corresponding to black individuals to create a dataset.
Using the likelihood of default, we also sample a binary outcome for each individual in the dataset and use it as the class label for the classification task.
The FICO dataset code and builder are available under the MIT License.

\noindent
\textbf{Adult dataset.} The list of features for the Adult dataset are age, class of worker, educational attainment, marital status, occupation, place of birth, usual hours worked per week past 12 months, and gender.
For feature descriptions, we refer the reader to \citet{ding2021retiring}. 
The Adult dataset is available under the MIT License
and was created using the  American Community Survey (ACS).

We use the SLSQP function in the Python scipy package to solve our constrained optimization problems for Adult dataset.
The parameters used for the SLSQP function are: number of iterations $=100$, ftol $=$ 1e-3, eps $=$ 1e-3.
All variables, other than class label and protected attribute, are normalized using the mean and standard deviation of the columns of the training partition.
For the implementation of \cite{rezaei2020fairness} classifier, we use the publicly-available code for their algorithm - \url{https://github.com/arezae4/fair-logloss-classification}. The parameter $C$ in their algorithm is set to be 0.005, as recommended in their code.

\begin{figure*}
    \centering
    \small
    \includegraphics[height=1in]{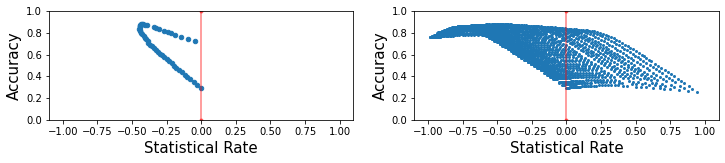}
    \subfloat[$\tau_0 {=}\tau_1 {\in}\set{1, \dots, 100}$]{\hspace{.5\linewidth}}
    \subfloat[$(\tau_0, \tau_1) {\in}\set{1, \dots, 100}^2$]{\hspace{.5\linewidth}}
    \caption{\small Accuracy and statistical rate of all classifiers for the FICO dataset for sub-population condition $\psi_{sr}$.
    Plot (a) presents statistical rate $H(f, \psi_{sr})$ vs accuracy for single-threshold classifiers. 
     Plot (b) present $H(f, \psi_{sr})$ vs accuracy for classifiers that use group-specific thresholds. 
    }
    \label{fig:fico_sr_acc_results}
\end{figure*}

\begin{figure*}
    \centering
    \small
    \includegraphics[height=1in]{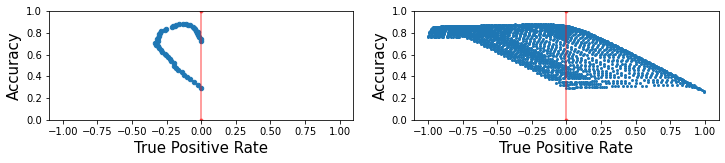}
    \subfloat[$\tau_0 {=}\tau_1 {\in}\set{1, \dots, 100}$]{\hspace{.5\linewidth}}
    \subfloat[$(\tau_0, \tau_1) {\in}\set{1, \dots, 100}^2$]{\hspace{.5\linewidth}}
    \caption{\small Accuracy and true positive rate of all classifiers for the FICO dataset for sub-population condition $\psi_{tpr}$.
    Plot (a) presents true positive rate $H(f, \psi_{tpr})$ vs accuracy for single-threshold classifiers. 
     Plot (b) present $H(f, \psi_{tpr})$ vs accuracy for classifiers that use group-specific thresholds. 
    }
    \label{fig:fico_tpr_acc_results}
\end{figure*}

\begin{figure*}[t]
    \centering
    \includegraphics[width=\linewidth]{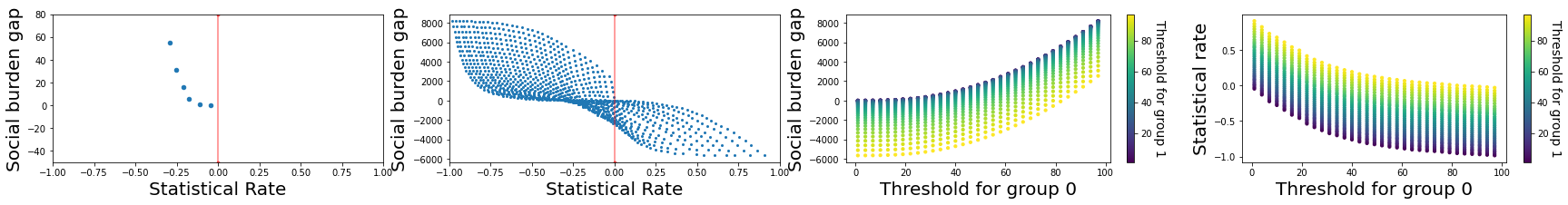}
    \subfloat[$\tau_0{=}\tau_1 \in \set{1, \dots, 100}$]{\hspace{.25\linewidth}}
    \subfloat[$(\tau_0, \tau_1){\in}\set{1, \dots, 100}^2$]{\hspace{.25\linewidth}}
    \subfloat[Social burden gap vs $\tau_0$]{\hspace{.25\linewidth}}
    \subfloat[Statistical rate vs $\tau_0$]{\hspace{.25\linewidth}}
    \caption{\small Statistical rate and social burden gap of all classifiers for the FICO dataset for quadratic cost function. Each point represents a classifier and the axes plot different properties of these classifiers.
    Plot (a) presents social burden gap $G(f, \psi_{sr})$ vs statistical rate $H(f, \psi_{sr})$ for classifiers that use the same threshold for both groups. 
     Plots (b), (c), (d) present social burden gap $G(f, \psi_{sr})$ vs statistical rate $H(f, \psi_{sr})$ for classifiers that can use group-specific threshold. 
    }
    \label{fig:fico_sr_results_quad}
\end{figure*}

\begin{figure*}
    \centering
    \small
    \includegraphics[width=\linewidth]{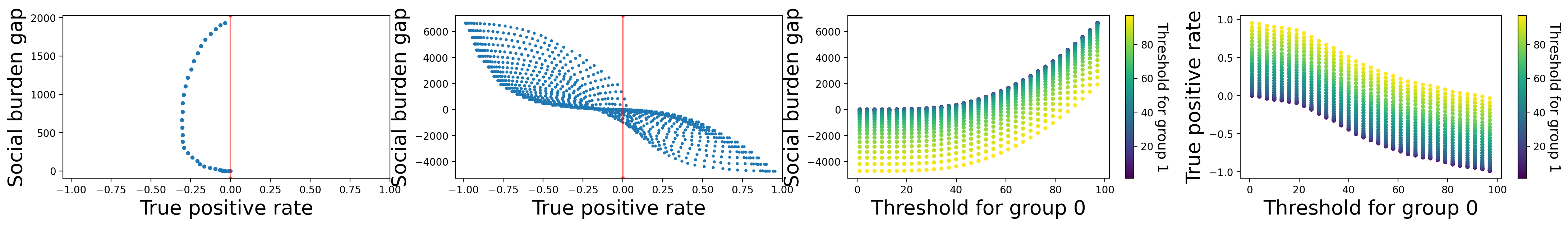}
    \subfloat[$\tau_0{=}\tau_1 \in \set{1, \dots, 100}$]{\hspace{.25\linewidth}}
    \subfloat[$(\tau_0, \tau_1){\in}\set{1, \dots, 100}^2$]{\hspace{.25\linewidth}}
    \subfloat[Social burden gap vs $\tau_0$]{\hspace{.25\linewidth}}
    \subfloat[True positive rate vs $\tau_0$]{\hspace{.25\linewidth}}
    \caption{\small True positive rate and social burden gap of all classifiers for the FICO dataset for quadratic cost function.
    Plot (a) presents social burden gap $G(f, \psi_{tpr})$ vs true positive rate $H(f, \psi_{tpr})$ for single-threshold classifiers.
     Plot (b), (c), (d) present $G(f, \psi_{tpr})$ vs $H(f, \psi_{tpr})$ for classifiers that use group-specific thresholds. 
    }
    \label{fig:fico_tpr_results_quad}
\end{figure*}

\begin{figure*}[t]
    \centering
    \small
    \includegraphics[width=\linewidth]{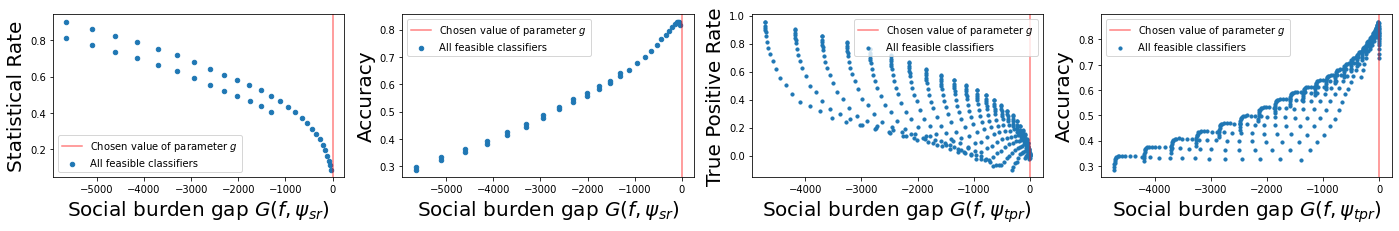}
    \subfloat[Stat. rate vs $G(f, \psi_{sr})$]{\hspace{.25\linewidth}}
    \subfloat[Accuracy vs $G(f, \psi_{sr})$]{\hspace{.25\linewidth}}
    \subfloat[TPR vs $G(f, \psi_{tpr})$]{\hspace{.25\linewidth}}
    \subfloat[Accuracy vs $G(f, \psi_{tpr})$]{\hspace{.25\linewidth}}
    \caption{\small Performance of classifiers that satisfy the modified fairness constraints.
    Plot (a,b) present statistical rate and accuracy vs social burden gap $G(f, \psi_{sr})$ for all classifiers for which condition~(\ref{eq}) is satisfied with $\psi = \psi_{sr}$ for quadratic cost function.
    Plot (c,d) present true positive rate and accuracy vs social burden gap $G(f, \psi_{tpr})$ for all classifiers for which condition~(\ref{eq}) is satisfied with $\psi = \psi_{tpr}$ for quadratic cost function.
    }
    \label{fig:fico_fair_clf_quad}
\end{figure*}

\section{Additional results for FICO} \label{sec:fico_appendix}

\nrev{
\paragraph{Accuracy of all classifiers vs $H(\cdot, \psi_{sr})$ and $H(\cdot, \psi_{tpr})$.}
We first expand the results presented in Section~\ref{sec:fico_main} and present the accuracy of all classifier vs statistical rate and true positive rate for sub-population conditions $\psi_{sr}$ and $\psi_{tpr}$ respectively.
Figure~\ref{fig:fico_sr_acc_results} presents the results for $\psi_{sr}$ when (a) classifier uses same thresholds for both groups and (b) when classifier uses different thresholds for both groups.
As expected, due to feature bias in the dataset, optimal classifier accuracy is higher when $H(\cdot, \psi_{sr})$ is less than 0, i.e., when the majority group individuals are selected a higher rate.
However, Figure~\ref{fig:fico_sr_acc_results}b also shows that there exist classifiers with $H(\cdot, \psi_{sr})$ greater than 0 (i.e., favoring the minority group) for which the loss in accuracy, compared to the classifier with maximum accuracy, is minimal.
Similar observations hold for the setting when sub-population condition is $\psi_{tpr}$; the results for this setting are presented in Figure~\ref{fig:fico_tpr_acc_results}.
}

\paragraph{Results with quadratic strategic cost function.}
We also provide additional results for the setting when cost function is quadratic and separable, i.e.
\[c(x,x') = \begin{cases} (x'^2-x^2),  & \text{ if} x' > x,\\ 0 & \text{otherwise}. \end{cases}\]
The classifiers considered here use group specific thresholds in the set $(\tau_0, \tau_1) \in \set{1, 3, \dots, 100}\times \set{1, 3, \dots, 100}$.
Figure~\ref{fig:fico_sr_results_quad} plots the statistical rate and social burden gap of all classifiers 
and Figure~\ref{fig:fico_tpr_results_quad} plots the true positive rate and social burden gap of all classifiers.
Once again we see that group-specific thresholds can lead to lower social burden gap than single-threshold classifiers.

Figure~\ref{fig:fico_fair_clf_quad} presents the performance of classifiers that constrain the upper bound in Theorem 3.4 to be less than 0.
In this case, there are fewer classifiers that are feasible with respect to the modified constraints (compared to the linear case) since the bounds on cost gradient $g_u$ (-2) and $g_l$ (-100) are relatively looser.

\paragraph{Resources used to run experiments.} All experiments were run on a MacBook system with 1.8 GHz processor and 8GB RAM.

\section{Additional results for Adult} \label{sec:adult_appendix}

\paragraph{Additional results.}

Recall that we defined, $f_{sr}$ as $\\$ $f_{sr} := \arg \min_f \E[L(f; X,Y)]$ subject to $H(f, \psi_{sr}) \geq \epsilon$ for pre-specified $\epsilon \in [-1,1]$.
Here $\epsilon$ denotes the desired statistical rate.
Figure~\ref{fig:adult_stat_rate_plot} plots the accuracy and social burden gap vs statistical rate of classifiers $f_{sr}$ for different $\epsilon$ values in the set $\set{-0.5, -0.4, \dots, 0.4, 0.5}$.
As we can see from the plots, social burden gap and accuracy decrease with increasing statistical rate.
However, for classifiers with statistical rate 0, average social burden gap is still greater than 0.

For a pre-specified $g \in \R$ we define $f_{strat}$ as $\\$
$f_{strat} := \arg \min_f \E[L(f; X,Y)]$ subject to 
$-\frac{1}{w_1^\star} \left(v_1 H(f, \psi) \right) - \delta \leq g$ (quantities $w_1^\star, v_1, \delta$ are defined in Theorem 4.1).
Here $g$ denotes the desired social burden gap.
Figure~\ref{fig:adult_social_burden_plot} plots the accuracy and statistical rate vs social burden gap of classifiers $f_{strat}$ for different $g$ values.

\begin{figure*}
    \centering
    \includegraphics[width=\linewidth]{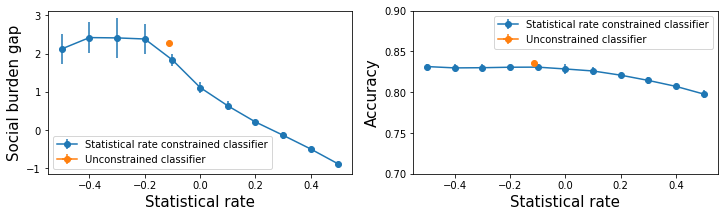}
    \caption{Results for Adult dataset. We plot the performance of unconstrained and fairness constrained classifiers. For the statistical rate constrained classifiers, we vary the desired statistical rate parameter and measure the social burden gap and accuracy of the resulting classifiers.}
    \label{fig:adult_stat_rate_plot}
\end{figure*}

\begin{figure*}
    \centering
    \includegraphics[width=\linewidth]{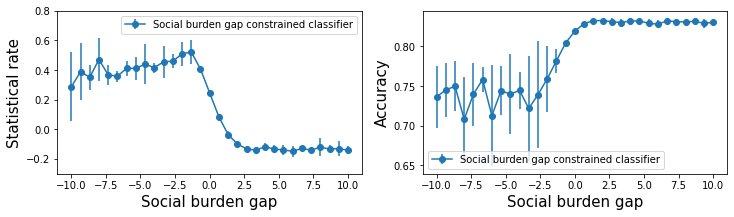}
    \caption{Results for Adult dataset. We plot the performance of unconstrained and fairness constrained classifiers. For the social burden gap constrained classifiers, we vary the desired social burden gap and measure the statistical rate and accuracy of the resulting classifiers.}
    \label{fig:adult_social_burden_plot}
\end{figure*}

\paragraph{Resources used to run experiments.} All experiments were run on a MacBook system with 1.8 GHz processor and 8GB RAM.

\end{document}